\documentclass[a4paper,UKenglish]{lipics-v2019}
\pdfoutput=1
\nolinenumbers

\usepackage{amsthm}

\usepackage[T1]{fontenc} 
\usepackage{amssymb}
\usepackage{enumerate}
\usepackage{float}
\usepackage{bookmark}
\usepackage{hyperref}
\usepackage{tikz}
\usepackage{comment}
\usepackage{wrapfig}

\title{Quantum Lower and Upper Bounds for 2D-Grid and Dyck Language} 

\def\luaffil{Center for Quantum Computer Science, Faculty of Computing, University of Latvia}

\author{Andris Ambainis, Kaspars Balodis, J\={a}nis Iraids}{\luaffil}{}{}{}

\author{Kamil Khadiev}{Kazan Federal University, Kazan, Russia}{}{}{}

\author{Vladislavs K\c{l}evickis, Kri\v{s}j\={a}nis Pr\={u}sis}{\luaffil}{}{}{}

\author{Yixin Shen}{Universit\'{e} de Paris, CNRS, IRIF, F-75006 Paris, France}{}{}{}

\author{Juris Smotrovs, Jevg\={e}nijs Vihrovs}{\luaffil}{}{}{}

\authorrunning{A. Ambainis et al.} 

\Copyright{Andris Ambainis, Kaspars Balodis, J\={a}nis Iraids, Kamil Khadiev, Vladislavs K\c{l}evickis, Kri\v{s}j\={a}nis Pr\={u}sis, Yixin Shen, Juris Smotrovs, Jevg\={e}nijs Vihrovs} 



\category{} 

\relatedversion{} 

\supplement{}

\funding{Supported by QuantERA ERA-NET Cofund in Quantum Technologies
implemented within the European Union's Horizon 2020 Programme (QuantAlgo project)
and ERDF project 1.1.1.5/18/A/020 ``Quantum algorithms: from complexity theory to experiment''. The research was funded by the subsidy allocated to Kazan Federal University for the state assignment in the sphere of scientific activities.}

\hideLIPIcs  

\EventEditors{Javier Esparza and Daniel Kr{\'a}l'}
\EventNoEds{2}
\EventLongTitle{45th International Symposium on Mathematical Foundations of Computer Science (MFCS 2020)}
\EventShortTitle{MFCS 2020}
\EventAcronym{MFCS}
\EventYear{2020}
\EventDate{August 24--28, 2020}
\EventLocation{Prague, Czech Republic}
\EventLogo{}
\SeriesVolume{170}
\ArticleNo{3}

\usetikzlibrary{arrows.meta}
\usetikzlibrary{math}

\usetikzlibrary{calc}
\usetikzlibrary{decorations.pathreplacing}
\usetikzlibrary{shapes}
\usetikzlibrary{arrows}
\usetikzlibrary{decorations}
\usetikzlibrary{decorations.pathmorphing}
\usetikzlibrary{patterns}
\usetikzlibrary{positioning}
\usetikzlibrary{fpu} 
\usepackage{algorithm}
\usepackage[noend]{algpseudocode}

\let\originalleft\left
\let\originalright\right
\renewcommand{\left}{\mathopen{}\mathclose\bgroup\originalleft}
\renewcommand{\right}{\aftergroup\egroup\originalright}
\newcommand{\ket}[1]{\left| #1 \right\rangle}

\newcommand{\sign}{\operatorname{sign}}
\newcommand{\NULL}{\operatorname{NULL}}

\newcommand{\FindFrom}[1]{\textsc{FindAtLeftmost}_{#1}}
\newcommand{\FindFromRight}[1]{\textsc{FindAtRightmost}_{#1}}
\newcommand{\FindFixedLen}[1]{\textsc{FindFixedLen}_{#1}}
\newcommand{\FindAny}[1]{\textsc{FindAny}_{#1}}
\newcommand{\FindFirst}[1]{\textsc{FindFirst}_{#1}}
\newcommand{\FindLeftFirst}[1]{\textsc{FindLeftFirst}_{#1}}
\newcommand{\FindRightFirst}[1]{\textsc{FindRightFirst}_{#1}}
\newcommand{\FindFixedPos}[1]{\textsc{FindFixedPos}_{#1}}

\makeatletter
\usetikzlibrary{decorations,backgrounds}
\def\pgfdecoratedcontourdistance{0pt}
\pgfset{
  decoration/contour distance/.code=%
    \pgfmathsetlengthmacro\pgfdecoratedcontourdistance{#1}}
\pgfdeclaredecoration{contour lineto closed}{start}{%
  \state{start}[
    next state=draw,
    width=0pt,
    persistent precomputation=\let\pgf@decorate@firstsegmentangle\pgfdecoratedangle]{%
    \pgfpathmoveto{\pgfpointlineattime{.5}
      {\pgfqpoint{0pt}{\pgfdecoratedcontourdistance}}
      {\pgfqpoint{\pgfdecoratedinputsegmentlength}{\pgfdecoratedcontourdistance}}}%
  }%
  \state{draw}[next state=draw, width=\pgfdecoratedinputsegmentlength]{%
    \ifpgf@decorate@is@closepath@%
      \pgfmathsetmacro\pgfdecoratedangletonextinputsegment{%
        -\pgfdecoratedangle+\pgf@decorate@firstsegmentangle}%
    \fi
    \pgfmathsetlengthmacro\pgf@decoration@contour@shorten{%
      -\pgfdecoratedcontourdistance*cot(-\pgfdecoratedangletonextinputsegment/2+90)}%
    \pgfpathlineto
      {\pgfpoint{\pgfdecoratedinputsegmentlength+\pgf@decoration@contour@shorten}
      {\pgfdecoratedcontourdistance}}%
    \ifpgf@decorate@is@closepath@%
      \pgfpathclose
    \fi
  }%
  \state{final}{}%
}
\makeatother
\tikzset{
  contour/.style={
    decoration={
      name=contour lineto closed,
      contour distance=#1
    },
    decorate}}
    
\hypersetup{
  pdftitle = {},
  pdfauthor = {Andris Ambainis, Kaspars Balodis, Janis Iraids, Kamil Khadiev, Vladislavs Klevickis, Krisjanis Prusis, Yixin Shen, Juris Smotrovs, Jevgenijs Vihrovs},
  pdfkeywords = {},
  pdfstartview=FitH,
  unicode=true
}

\newcommand{\ex}[3]{\ensuremath{\textsc{Ex}_{#1}^{#2\mid#3}}}
\newcommand{\exnn}{\ex{2m}{m}{m+1}}

\DeclareMathOperator*{\orf}{\textsc{Or}}
\DeclareMathOperator*{\andf}{\textsc{And}}
\DeclareMathOperator{\dyck}{\textsc{Dyck}}
\newcommand{\im}{f}

\newcommand{\dirtwodtext}{\textsc{2D-DConnectivity}}
\DeclareMathOperator{\dirtwod}{\dirtwodtext}
\newcommand{\undirtwodtext}{\textsc{2D-Connectivity}}
\DeclareMathOperator{\undirtwod}{\undirtwodtext}
\DeclareMathOperator{\dirdd}{d\textsc{D-DConnectivity}}
\DeclareMathOperator{\undirdd}{\textsc{}d\textsc{D-Connectivity}}

\theoremstyle{plain}
\newtheorem{cor}{Corollary}

\begin{document}

\maketitle

\begin{abstract}
We study the quantum query complexity of two problems.

First, we consider the problem of determining if a sequence of parentheses is a properly balanced one ({\em a Dyck word}), with a depth of at most $k$. We call this the $\dyck_{k,n}$ problem. We prove a lower bound of $\Omega(c^k \sqrt{n})$, showing that the complexity of this problem increases exponentially in $k$. Here $n$ is the length of the word. When $k$ is a constant, this is interesting as a representative example of star-free languages for which a surprising $\tilde{O}(\sqrt{n})$ query quantum algorithm was recently constructed by Aaronson et al. \cite{Aaronson2018AQQ}. Their proof does not give rise to a general algorithm.
When $k$ is not a constant, $\dyck_{k,n}$ is not context-free. We give an algorithm with $O\left(\sqrt{n}(\log{n})^{0.5k}\right)$ quantum queries for $\dyck_{k,n}$ for all $k$. This is better than the trival upper bound $n$ for $k=o\left(\frac{\log(n)}{\log\log n}\right)$. 

Second, we consider connectivity problems on grid graphs in 2 dimensions, if some of the edges of the grid may be missing. By embedding the ``balanced parentheses'' problem into the grid, we show a lower bound of $\Omega(n^{1.5-\epsilon})$ for the directed 2D grid and
$\Omega(n^{2-\epsilon})$ for the undirected 2D grid.
The directed problem is interesting as a black-box model for a class of classical dynamic programming strategies including the one that is usually used for the well-known edit distance problem. 
We also show a generalization of this result to more than 2 dimensions. 
\end{abstract}

\section{Introduction}
\label{s:intro}

We study the quantum query complexity of two problems:

{\bf Quantum complexity of regular languages.} 
Consider the problem of recognizing whether an $n$-bit string belongs to a given regular language. This models a variety of computational tasks that can be described by regular languages.
In the quantum case, the most commonly used model for studying the complexity of various problems is the query model. For this setting, Aaronson, Grier and Schaeffer \cite{Aaronson2018AQQ} recently showed that any regular language $L$ has one of three possible quantum query complexities on inputs of length $n$: 
$\Theta(1)$ if the language can be decided by looking at $O(1)$ first or last symbols of the word;
$\tilde{\Theta}(\sqrt{n})$ if the best way to decide $L$ is Grover's search (for example, for the language consisting of all words containing at least one letter a);
$\Theta(n)$ for languages in which we can embed counting modulo some number $p$ which has quantum query complexity $\Theta(n)$.

As shown in \cite{Aaronson2018AQQ}, a regular language being of complexity $\tilde{O}(\sqrt{n})$ (which includes the first two cases above) is equivalent to it being star-free. Star-free languages are defined as the languages which have regular expressions not containing the Kleene star (if it is allowed to use the complement operation). Star-free languages are one of the most commonly studied subclasses of regular languages and there are many equivalent characterizations of them.
%
%
One class of the star-free languages mentioned in \cite{Aaronson2018AQQ} is the Dyck languages (with one type of parenthesis) with constant height $k$. Dyck language with height $k$ consists of words with balanced number of parentheses such that in no prefix the number of opening parentheses exceeds the number of closing parentheses by more than $k$; we denote the problem of determining if an input of length $n$ belongs to this language by $\dyck_{k,n}$. In case of unbounded height $k=\frac{n}{2}$, the language is a fundamental example of a context-free language that is not regular. 
 When more types of parenthesis are allowed, the famous Chomsky–Schützenberger representation theorem shows that any context-free language is the homomorphic image of the intersection of a Dyck language and a regular language.

{\bf Our results.} We show that an exponential dependence of the complexity on $k$ is unavoidable. Namely, for the balanced parentheses language, we have
\begin{itemize}
\item
there exists $c>1$ such that, for all $k\leq \log n$, the quantum query complexity is $\Omega(c^k \sqrt{n})$;
\item
If $k=c\log n$ for an appropriate constant $c$, the quantum query complexity is $\Omega(n^{1-\epsilon})$. 
\end{itemize}

Thus, the exponential dependence on $k$ is unavoidable and distinguishing sequences of balanced parentheses of length $n$ and depth $\log n$ is almost as hard as distinguishing sequences of length $n$ and arbitrary depth. 

Similar lower bounds have recently been independently proven by Buhrman et al. \cite{buhrman2019quantum}.

Additionally, we give an explicit algorithm (see Theorem~\ref{th:upper_bound}) for the decision problem $\dyck_{k,n}$ with $O\left(\sqrt{n}(\log{n})^{0.5k}\right)$ quantum queries. The algorithm also works when $k$ is not a constant and is better than the trivial upper bound of $n$ when $k=o\left(\frac{\log(n)}{\log\log n}\right)$. 

{\bf Finding paths on a grid.}
The second problem that we consider is graph connectivity on subgraphs of the 2D grid. Consider a 2D grid with vertices $(i, j)$,
$i\in \{0, 1, \ldots, n\}, j\in \{0, 1, \ldots, k\}$ and edges from $(i, j)$ to $(i+1, j)$ and $(i, j+1)$. The grid can be either directed (with edges in the directions of increasing coordinates) or undirected. We are given an unknown subgraph $G$ of the 2D grid and we can perform queries to variables $x_u$ (where $u$ is an edge of the grid) defined by $x_u=1$ 
if $u$ belongs to $G$ and 0 otherwise.
The task is to determine whether $G$ contains a path from $(0, 0)$ to $(n, k)$. 

Our interest in this problem is driven by the edit distance problem. In the edit distance problem, we are given two strings $x$ and $y$ and have to determine the smallest number of operations (replacing one symbol by another, removing a symbol or inserting a new symbol) with which one can transform $x$ to $y$. If $|x|\leq n, |y|\leq k$,
the edit distance is solvable in time $O(nk)$ by dynamic programming \cite{wagner1974string}. If $n=k$ then, under the strong exponential time hypothesis (SETH), there is no classical algorithm computing edit distance in time $O(n^{2-\epsilon})$ for $\epsilon>0$ \cite{backurs2015edit} and the dynamic programming algorithm is essentially optimal.

However, SETH does not apply to quantum algorithms.
Namely, SETH asserts that there is no algorithm for general instances of SAT that is substantially better than naive search. Quantumly, a simple use of Grover's search gives a quadratic advantage over naive search. This leads to the question: can this quadratic advantage be extended to edit distance (and other problems that have lower bounds based on SETH)?

Since edit distance is quite important in classical algorithms, the question about its quantum complexity has attracted a substantial interest from various researchers. Boroujeni et al. \cite{boroujeni2018approximating} invented a better-than-classical quantum algorithm 
for approximating the edit distance which was later superseded by a better classical algorithm of \cite{C+18}. However, there has been no quantum algorithms computing the edit distance exactly (which is the most important case).

The main idea of the classical algorithm for edit distance is as follows:
\begin{itemize}
\item
We construct a weighted version of the directed 2D grid (with edge weights 0 and 1) that encodes the edit distance problem for strings $x$ and $y$, with the edit distance being equal to the length of the shortest directed path from $(0, 0)$ to $(n, k)$.
\item
We solve the shortest path problem on this graph and obtain the edit distance.
\end{itemize}
As a first step, we can study the question of whether the shortest path is of length 0 or more than 0. Then, we can view edges of length 0 as present and edges of length 1 as absent. The question ``Is there a path of length of 0?'' then becomes ``Is there a path from $(0, 0)$ to $(n, k)$ in which all edges are present?''. A lower bound for this problem would imply a similar lower bound for the shortest path problem and a quantum algorithm for it may contain ideas that would be useful for a shortest path quantum algorithm.

{\bf Our results.}
We use our lower bound on the balanced parentheses language to show an $\Omega(n^{1.5-\epsilon})$ lower bound for the connectivity problem on the directed 2D grid. This shows a limit on quantum algorithms for finding edit distance through the reduction to shortest paths. More generally, for an $n\times k$ grid ($n>k$), our proof gives a lower bound of
$\Omega((\sqrt{n}k)^{1-\epsilon})$.

The trivial upper bound is $O(nk)$ queries, since there are $O(nk)$ variables. There is no nontrivial quantum algorithm, except for the case when $k$ is very small. Then, we show that the connectivity problem can be solved with $O(\sqrt{n} \log^{k/2} n)$ quantum queries\footnote{Aaronson et al. \cite{Aaronson2018AQQ} also give a bound of $O(\sqrt{n} \log^{m-1} n)$ but in this case $m$ is the rank of the syntactic monoid which can be exponentially larger than $k$.} but this bound becomes trivial already for $k=\Omega(\frac{\log n}{\log \log n})$.

For the undirected 2D grid, we show a lower bound of $\Omega((nk)^{1-\epsilon})$, whenever $k\geq \log n$. Thus, the naive algorithm is almost optimal in this case. We also extend both of these results to higher dimensions, obtaining a lower bound of 
$\Omega((n_1 n_2 \ldots n_d)^{1-\epsilon})$ for an undirected $n_1 \times n_2 \times \ldots \times n_d$ grid in $d$ dimensions and a lower bound of 
$\Omega(n^{(d+1)/2-\epsilon})$ for a directed $n\times n \times \ldots \times n$ grid in $d$ dimensions.

In a recent work, an $\Omega(n^{1.5})$ lower bound for edit distance was shown by Buhrman et al. \cite{buhrman2019quantum}, 
assuming a quantum version of the Strong Exponential Time hypothesis (QSETH). As part of this result they give an $\Omega(n^{1.5})$ query lower bound for a different path problem on a 2D grid. Then QSETH is invoked to prove
that no quantum algorithm can be faster than the best algorithm for this shortest path problem. 
Neither of the two
results follow directly one from another, as different shortest path problems are used.

\section{Definitions}
\label{s:defs}

For a word $x\in\Sigma^*$ and a symbol $a\in\Sigma$, let $|x|_a$ be the number of occurrences of $a$ in $x$.

For two (possibly partial) Boolean functions $g: G \rightarrow \{0,1\}$, where $G \subseteq \{0,1\}^n$, and $h:H\rightarrow \{0,1\}$, where $H\subseteq \{0,1\}^m$, we define the composed function $g\circ h: D\rightarrow \{0,1\}$, with $D\subseteq \{0,1\}^{nm}$, as
$\left(g\circ h\right) (x) = g\left(h(x_1,\dots,x_m), \dots, h(x_{(n-1)m+1},\dots, x_{nm})\right).$
Given a Boolean function $f$ and a nonnegative integer $d$, we define $f^d$ recursively as $f$ iterated $d$ times: $f^d=f\circ f^{d-1}$ with $f^1=f$.

For a matrix $\Gamma$, $\|\Gamma\|$ denotes the spectral norm of $\Gamma$: $\|\Gamma\|= \max_{\overrightarrow{x}\neq 0}\frac{\|\Gamma \overrightarrow{x}\|}{\|\overrightarrow{x}\|}$
where $\|\overrightarrow{x}\|$ is the $2$-norm of a vector.

{\bf Quantum query model.}
    We use the standard form of the quantum query model. 
    Let $f:D\rightarrow \{0,1\},D\subseteq \{0,1\}^n$ be an $n$ variable function we wish to compute on an input $x\in D$. We have an oracle access to the input $x$ --- it is realized by a specific unitary transformation usually defined as $\ket{i}\ket{z}\ket{w}\rightarrow \ket{i}\ket{z+x_i\pmod{2}}\ket{w}$ where the $\ket{i}$ register indicates the index of the variable we are querying, $\ket{z}$ is the output register, and $\ket{w}$ is some auxiliary work-space. An algorithm in the query model consists of alternating applications of arbitrary unitaries independent of the input and the query unitary, and a measurement in the end. The smallest number of queries for an algorithm that outputs $f(x)$ with probability $\geq \frac{2}{3}$ on all $x$ is called the quantum query complexity of the function $f$ and is denoted by $Q(f)$.
    
    Let a symmetric matrix $\Gamma$ be called an adversary matrix for $f$ if the rows and columns of $\Gamma$ are indexed by inputs $x\in D$ and $\Gamma_{xy}=0$ if $f(x)=f(y)$. Let $\Gamma^{(i)}$ be a similarly sized matrix such that $\displaystyle \Gamma^{(i)}_{xy}=\begin{cases}\Gamma_{xy}&\text{ if }x_i\neq y_i\\ 0&\text{ otherwise}\end{cases}$. Then let 
    $\displaystyle Adv^{\pm}(f)=\max_{\substack{\Gamma\text{ - an adversary}\\ \text{matrix for }f}}{\frac{\|\Gamma\|}{\max_i{\|\Gamma^{(i)}\|}}}$
    be called the adversary bound and let
    $\displaystyle Adv(f)=\max_{\substack{\Gamma\text{ - an adversary matrix for }f\\ \Gamma \text{ - nonnegative}}}{\frac{\|\Gamma\|}{\max_i{\|\Gamma^{(i)}\|}}}$
    be called the positive adversary bound.
        The following facts will be relevant for us:
        $Adv(f)\leq Adv^\pm(f)$;
        $Q(f)=\Theta(Adv^{\pm}(f))$ \cite{Reichardt11};
        $Adv^{\pm}$ composes exactly even for partial Boolean functions $f$ and $g$, meaning, $Adv^\pm(f\circ g)=Adv^\pm(f)\cdot Adv^\pm(g)$ \cite[Lemma~6]{kimmel2012quantum}.
    
{\bf Reductions.}
    We will say that a Boolean function $f$ is reducible to $g$ and denote it by $f \leqslant g$ if there exists an algorithm that given an oracle $O_x$ for an input of $f$ transforms it into an oracle $O_y$ for $g$ using at most $O(1)$ calls of oracle $O_x$ such that $f(x)$ can be computed from $g(y)$. Therefore, from $f \leqslant g$ we conclude that $Q(f)\leq Q(g)$ because one can compute $f(x)$ using the algorithm for $g(y)$ and the reduction algorithm that maps $x$ to $y$.

{\bf Dyck languages of bounded depth.}
Let $\Sigma$ be an alphabet consisting of two symbols: \texttt{(} and \texttt{)}.
The Dyck language $L$ consists of all $x\in \Sigma^*$ that represent a correct sequence of opening and closing parentheses.
We consider languages $L_k$ consisting of all words $x\in L$ where the number of opening parentheses that are not closed yet never exceeds $k$.
The language $L_k$ corresponds to a query problem $\dyck_{k,n}(x_1, ..., x_n)$ where $x_1, \ldots, x_n \in \{0, 1\}$ describe a word of length $n$
in the natural way: the $i^{\rm th}$ symbol of $x$ is \texttt{(} if $x_i=0$ and \texttt{)} if $x_i=1$. $\dyck_{k,n}(x)=1$ iff the word $x$ belongs to $L_k$.
For all $x \in \{0,1\}^n$, we define  $f(x)=|x|_{\texttt{0}}-|x|_{\texttt{1}}$,
we call it the \textbf{balance}.
We define a $+k$-substring (resp. $-k$-substring) as a substring whose balance is equal to $k$ (resp. equal to $-k$).
A $\pm k-$substring is a substring whose balance is equal to $k$ in absolute
value.
For all $0\leq i\leq j\leq n-1$, we define $x[i,j]=x_i,x_{i+1},\cdots ,x_j$. Finally, we define $h(x)= \max_{0\leq i\leq n-1}f(x[0,i])$ and $h^-(x)= \min_{0\leq i\leq n-1}f(x[0,i])$.
%
A substring $x[i,j]$  is \emph{minimal} if it does not contain a substring $x[i',j']$ such that $(i,j)\neq (i',j')$, and $f(x[i',j'])=f(x[i,j])$.

{\bf Connectivity on a directed 2D grid.} Let $G_{n,k}$ be a directed version of an $n\times k$ grid in two dimensions, with vertices $(i, j), i\in \{0,1,\ldots, n\}, j\in \{0,1,\ldots, k\}$
and directed edges from $(i, j)$ to $(i+1, j)$ (if $i<n$) and from $(i, j)$ to $(i, j+1)$ (if $j<k$).
If $G$ is a subgraph of $G_{n,k}$, we can describe it by variables $x_e$ corresponding to edges $e$ of $G_{n,k}$: $x_e=1$ 
if the edge $e$ belongs to $G$ and $x_e=0$ otherwise. We consider a problem $\dirtwod$ in which one has to determine if $G$ contains a 
path from $(0, 0)$ to $(n, k)$: 
$\dirtwod_{n,k}(x_1, \ldots, x_m)=1$ (where $m$ is the number of edges in $G_{n, k}$)
iff such a path exists. 

{\bf Connectivity on an undirected 2D grid.} Let $G_{n, k}$ be an undirected $n\times k$ grid and let $G$ be a subgraph of $G_{n, k}$. We describe $G$ by 
variables $x_e$ in a similar way and define $\undirtwod_{n, k}(x_1, \ldots, x_m)=1$ iff $G$ contains a path from $(0, 0)$ to $(n, k)$.
We also consider $d$ dimensional versions of these two problems, on $n_1\times n_2 \times \ldots n_d$ grids. 
In the directed version ($\dirdd$), we have a subgraph $G$ of a directed grid (with edges directed in the directions from $(0, \ldots, 0)$ to
$(n_1, \ldots, n_d)$) and 
$\dirdd(x_1, \ldots, x_m)=1$
iff $G$ contains a directed path from $(0, \ldots, 0)$ to
$(n_1, \ldots, n_d)$.
The undirected version is defined similarly, with an undirected grid instead of a directed one.

\section{A quantum algorithm for membership testing of \texorpdfstring{$\dyck_{k,n}$}{DYCK\_\{n,k\}}}

In this section, we give a quantum algorithm for $\dyck_{k,n}(x)$, where $k$ can be a function of $n$. The general idea is that $\dyck_{k,n}(x)=0$ if and only if one of the following conditions holds:
(i) $x$ contains a $+(k+1)$-substring;
    (ii) $x$ contains a substring $x[0,i]$ such that the balance $f(x[0,i])=-1$;
    (iii) the balance of the entire word $f(x)\neq 0$.
    
The main algorithm is presented in Section \ref{sec:dyck-algo}. It based on a subroutine presented in Section \ref{sec:substr}.


\subsection{\texorpdfstring{$\pm k$}{±k}-Substring Search algorithm}\label{sec:substr}


%
The goal of this section is to describe a quantum algorithm which searches for a substring $x[i,j]$ that has a balance $f(x[i,j])\in\{+k,-k\}$ for some integer $k$. 
Throughout this section, we find and consider only \textbf{minimal} substrings.  A substring is minimal if it does not contain a proper substring with the same balance. Throughout this section we use the following easily verifiable facts:
\begin{itemize}
    \item For any two minimal $\pm k$-substrings $x[i,j]$ and $x[k,l]$: $i<k \implies j<l$. This induces a natural linear order among all $\pm k$-substrings according to their starting (or, equivalently, ending) positions.
    \item Minimal $+k$-substrings do not intersect with minimal $-k$-substrings.
    \item If $x[l_1,r_1]$ and $x[l_2,r_2]$ with $l_1<l_2$ are two \textbf{consecutive} minimal $(k-1)$-substrings and their signs are the same, then $x[l_1,r_2]$ is a $k$-substring with this sign.
\end{itemize} 
This algorithm is the basis of our algorithms for $\dyck_{k,n}$.
The algorithm works in a recursive way. It searches for two consecutive minimal $\pm (k-1)$-substrings $x[l_1,r_1]$ and $x[l_2,r_2]$  such that they either overlap or there are no $\pm (k-1)$-substrings between them. If both substrings $x[l_1,r_1]$ and $x[l_2,r_2]$ are $+(k-1)$-substrings, then we get a minimal $+k$-substring in total. If both substrings are $-(k-1)$-substrings, then we get a minimal $-k$-substring in total.

Our algorithm utilizes three subroutines. The first one is $\FindFrom{k}(l,r,t,d,s) $ which accepts as inputs:
    the borders $l$ and $r$, where $l$ and $r$ are integers such that $0\leq l \leq r\leq n-1$;
    a position $t\in \{l,\dots, r\}$;
    a maximal length $d$ for the substring, where $d$ is an integer such that $0<d\leq r-l+1$;
    the sign of the balance $s\subseteq \{+1,-1\}$. $+1$ is used for searching for a $+k$-substring, $-1$ is used for searching for a $-k$-substring, $\{+1,-1\}$ is used for searching for both.  
%
It outputs a triple $(i,j,\sigma)$ such that $l\leq i\leq t\leq j\leq r$, $j-i+1\leq d$, $f(x[i,j])\in\{+k,-k\}$ and $\sigma=\sign(f(x[i,j]))\in s$. The substring should be the leftmost one that contains $t$, i.e. there is no other minimal $x[i',j']$ such that $i'<i$, $t\in[i',j']$, $f(x[i',j'])=f(x[i,j])$.   If no such substrings have been found, the algorithm returns $\NULL$.

The second one is $\FindFromRight{k}$. It is similar to the $\FindFrom{k}$, but finds the rightmost $\pm k$-substring, i.e. there is no other minimal $x[i',j']$ such that $j'>j$, $t\in[i',j']$, $f(x[i',j'])=f(x[i,j])$

The third one is $\FindFirst{k}(l,r,s,direction)$ and accepts as inputs:
    the borders $l$ and $r$, where $l$ and $r$ are integers such that $0\leq l \leq r\leq n-1$;
    the sign of the balance $s\subseteq \{+1,-1\}$. 
    a $direction \in \{left,right\}$. 
When the direction is right (respectively left), $\FindFirst{k}$ finds the first $\pm k$-substring from the left to the right (respectively from the right to the left) in $[l,r]$ of sign $s$.

These three subroutines are interdependent since $\FindFrom{k}$ uses $\FindFirst{k-1}$ and $\FindFromRight{k-1}$ as subroutines, $\FindFirst{k}$ uses $\FindFrom{k}$ and $\FindFromRight{k}$ as subroutines. 
A description of $\FindFrom{k}(l,r,t,d,s)$  follows. The algorithm is presented in Appendix~\ref{sec:desc_k_substring_base}.  The description of the subroutine $\FindFromRight{k}(l,r,t,d,s)$ is similar and is omitted.

When $k=2$, the procedure $\FindFrom{2}(l,r,t,d,s)$ checks that $x_t=x_{t-1}$ and $\sign(f(x[t-1,t]))\in s$. If yes, it has found the substring. Otherwise, it checks if $x_t=x_{t+1}$ and $\sign(f(x[t,t+1]))\in s$. If both checks fail, the procedure returns $\NULL$. For $k>2$ the procedure is the following.
\begin{description}
    \item[Step $1$.] Check whether $t$ is inside a $\pm (k-1)$-substring of length at most $d-1$, i.e. 
        
        $v=(i,j,\sigma)\gets  \FindFrom{k-1}(l,r,t,d-1,\{+1,-1\}).$
        If $v\neq \NULL$, then $(i_1,j_1,\sigma_1)\gets(i,j,\sigma)$ and the algorithm goes to Step $2$. Otherwise, the algorithm goes to Step $6$.
    \item[Step $2$.] Check whether $i_1-1$ is inside a $\pm (k-1)$-substring of length at most $d-1$ and choose the rightmost one:
    $v=(i,j,\sigma)\gets\FindFromRight{k-1}(l,r,i_1-1,d-1,\{+1,-1\}).$
    
        If $v=\NULL$, then the algorithm goes to Step $3$.
        If $v\neq \NULL$ and $\sigma=\sigma_1$, then $(i_2,j_2,\sigma_2)\gets(i,j,\sigma)$ and go to Step $8$. Otherwise, go to Step $4$.
    \item[Step $3$.] Search for the first $\pm (k-1)$-substring on the left from $i_1-1$ at distance at most $d$, i.e. 
    $v=(i,j,\sigma)\gets\FindFirst{k-1}(\min(l,j_1-d+1),i_1-1),\{+1,-1\},left).$
    If $v\neq \NULL$ and $\sigma_1=\sigma$, then $(i_2,j_2,\sigma_2)\gets(i,j,\sigma)$ and go to Step $8$. Otherwise, go to Step $4$.    
        
    \item[Step $4$.] Check whether $j_1+1$ is inside a $\pm (k-1)$-substring of length at most $d-1$, i.e. 
    
    $v=(i,j,\sigma)\gets\FindFrom{k-1}(l,r,j_1+1,d-1,\{+1,-1\}).$
    
        If $v\neq \NULL$, then $(i_2,j_2,\sigma_2)\gets(i,j,\sigma)$ and go to Step $8$. Otherwise, go to Step $5$.
        
    \item[Step $5$.] Search for the first $\pm (k-1)$-substring on the right from $j_1+1$ at distance at most $d$, i.e. 
    $v=(i,j,\sigma)\gets\FindFirst{k-1}(j_1+1,\min(i_1+d-1,r),\{+1,-1\},right).$ 
    
    If $v\neq \NULL$, then $(i_2,j_2,\sigma_2)\gets(i,j,\sigma)$,
    then go to Step $8$. Otherwise, return $\NULL$.

    \item[Step $6$.] Search for the first $\pm (k-1)$-substring on the right at distance at most $d$ from $t$, i.e.
    
   $v=(i,j,\sigma)\gets\FindFirst{k-1}(t,\min(t+d-1,r),\{+1,-1\},right)$
   
    If $v\neq \NULL$, then $(i_1,j_1,\sigma_1)\gets(i,j,\sigma)$ and go to Step $7$. Otherwise, returns $\NULL$.
    \item[Step $7$.] Search for the first $\pm (k-1)$-substring on the left from $t$ at distance at most $d$, i.e. 
    
    $v=(i,j,\sigma)\gets\FindFirst{k-1}(\max(l,t-d+1),t),\{+1,-1\},left)$
    
    If  $v\neq \NULL$, then $(i_2,j_2,\sigma_2)\gets(i,j,\sigma)$ and go to Step $8$. Otherwise, returns $\NULL$.
    \item[Step $8$.] If $\sigma_1=\sigma_2$, $\sigma_1\in s$ and $\max(j_1,j_2)-\min(i_1,i_2)+1\leq d$ , the subroutine returns $(\min(i_1,i_2),\max(j_1,j_2),\sigma_1)$, otherwise returns $\NULL$. 
\end{description}

By construction and induction on $k$, the two $\pm(k-1)$-substrings $x[i_1,j_1]$ and $x[i_2,j_2]$ (if they exist) involved in the procedure $\FindFrom{k}$ are always consecutive and minimal. $\FindFrom{k}$ thus returns a $\pm k$-substring, if both substrings have the same sign.

Using this basic procedure, we then search for a $\pm k-$substring by searching for a $t$ and $d$ such that
$\FindFrom{k}(l,r,t,d,s)$
returns a non-$\NULL$ value. Unfortunately,
our algorithms have two-sided bounded error: they
can, with small probability, return $\NULL$ even if a substring exists or return a wrong substring instead
of $\NULL$. In this setting, Grover's search algorithm
is not directly applicable and we need to use a more
sophisticated search~\cite{Hoyer03}. Furthermore, simply applying
the search algorithm naively does not give the right complexity.
Indeed, if we search for a substring of length roughly $d$
(say between $d$ and $2d$), we can find one with expected running time $O(\sqrt{(r-l)/d})$
because at least $d$ values of $t$ will work.
On the other hand, if there are no such substrings, the expected running time will be 
$O(\sqrt{r-l})$. Intuitively, we can do better
because if there is a substring of length at least $d$ then there are at least $d$ values of $t$ that work. Hence, we only need to distinguish between no solutions, or at least $d$. This allows to stop the Grover iteration early and make
$O(\sqrt{(r-l)/d})$ queries in all cases.

\begin{lemma}[Modified from \cite{Hoyer03}, Appendix~\ref{sec:bounded}]\label{lem:modified_hoyer03}
    Given $n$ algorithms, quantum or classical, each computing some bit-value with bounded error probability, and some $T\geqslant1$,
    there is a quantum algorithm that uses $O(\sqrt{n/T})$
    queries and with constant probability:
        returns the index of a ``1'', if there are at least $T$
            ``1s'' among the $n$ values;
        returns $\NULL$ if there are no ``1'';
        returns anything otherwise.
\end{lemma}

The algorithm that uses above ideas is presented in Algorithm \ref{alg:fixedlenk}.

\begin{algorithm}[ht]
    \caption{$\FindFixedLen{k}(l,r,d,s)$. Search for any $\pm k$-substring of length $\in[d/2,d]$\label{alg:fixedlenk}}
    \begin{algorithmic}
        \State Find $t$ such that $v_t\gets\FindFrom{k}(l,r,t,d,s)\neq\NULL$ using Lemma~\ref{lem:modified_hoyer03} with $T=d/2$.
        \State \Return{$v_t$ or $\NULL$ if none}.
    \end{algorithmic}
\end{algorithm}


We can then write an algorithm $\FindAny{k}(l,r,s)$ that searches for any $\pm k$-substring. We consider a randomized algorithm that uniformly chooses a of power $2$ from $[2^{\lceil\log_2 k\rceil},  (r-l)]$, i.e.  $d\in\{2^{\lceil\log_2 k\rceil}, 2^{\lceil\log_2 k\rceil+1},\dots, 2^{\lceil\log_2 (r-l)\rceil}\}$. For the chosen $d$, we run Algorithm \ref{alg:fixedlenk}. So, the algorithm will succeed with probability at least $O(1/\log (r-l))$. We can apply Amplitude amplification and ideas from Lemma \ref{lem:modified_hoyer03} to this and get an algorithm that uses $O(\sqrt{\log (r-l)})$ iterations. 

\begin{algorithm}[ht]
    \caption{$\FindAny{k}(l,r,s)$. Search for any $\pm k$-substring.\label{alg:substringk_any}}
    \begin{algorithmic}
        \State Find $d\in\{2^{\lceil\log_2 k\rceil}, 2^{\lceil\log_2 k\rceil+1},\dots, 2^{\lceil\log_2 (r-l)\rceil}\}$ such that:\\ $v_d\gets\FindFixedLen{k}(l,r,d,s)\neq\NULL$ using amplitude amplification.
        \State \Return{$v_d$ or $\NULL$ if none}.
    \end{algorithmic}
\end{algorithm}

Finally, we present the algorithm that finds the first $\pm k$-substring -- $\FindFirst{k}$.  Let us consider the case $direction=right$. We first find the smallest segment from the left to the right such that its length $w$ is a power of $2$ and it contains a $\pm k$-substring. We do so by doubling the length of the segment until we find a $\pm k$-substring. 
We now have a segment that contains a $\pm k$-substring and we want to find the leftmost one. We do so by the following variant of binary search.
 At each step let $mid = \lfloor(lBorder+ rBorder)/2 \rfloor$ be the middle of the search segment $[lBorder, rBorder]$. There are three cases:
\begin{itemize}
\item There is a $k$-substring in $[lBorder, mid]$, then the leftmost $k$-substring is in this segment. 
\item There are no $k$-substrings in $[lBorder, mid]$, but $mid$ is inside a $k$-substring. Then the leftmost $k$-substring that contains $mid$ is the required substring.
\item There are no $k$-substrings in $[lBorder, mid]$ and $mid$ is not inside a $k$-substring. Then the required substring is in $[mid+1, rBorder]$.
\end{itemize}
Each iteration of the loop the algorithm halves the search space or finds the first $k$-substring itself if it contains $mid$.
If $direction=left$, we replace $\FindFrom{k}$ by $\FindFromRight{k}$ that finds the rightmost $\pm k$-substring that containts $mid$.
A detailed description of this algorithm is presented in Appendix~\ref{apx:findfirst}.

\begin{proposition}\label{pr:findfrom}
    For any $\varepsilon>0$ and $k$, algorithms $\FindFrom{k}$, $\FindFixedLen{k}$, $\FindAny{k}$ and $\FindFirst{k}$ have two-sided error probability $\varepsilon<0.5$ and return, when correct:
    \begin{itemize}
        \item If $t$ is inside a $\pm k-$substring of sign $s$ of length
            up to $d$ in $x[l,r]$, then
            $\FindFrom{k}$
            will return such a substring, otherwise it returns $\NULL$.
            The running time is $O(\sqrt{d}(\log (r-l))^{0.5(k-2)})$.
        \item $\FindFixedLen{k}$ either returns a
            $\pm k-$substring of sign $s$ and length at most $d$ in $x[l,r]$, or $\NULL$. It is only guaranteed to return a substring if there exists $\pm k-$substring of length at least $d/2$, otherwise it can return $\NULL$.
            The running time is $O(\sqrt{r-l}(\log (r-l))^{0.5(k-2)})$.
        \item $\FindAny{k}$ returns any
            $\pm k-$substring of sign $s$ in $x[l,r]$, otherwise it returns $\NULL$.
            The running time is $O(\sqrt{r-l}(\log (r-l))^{0.5(k-1)})$.
        \item $\FindFirst{k}$ returns the first $\pm k-$substring
            of sign $s$ in $x[l,r]$ in the specified direction, otherwise it returns $\NULL$.
            The running time is $O(\sqrt{r-l}(\log(r-l))^{0.5(k-1)})$.
    \end{itemize}
\end{proposition}

\begin{proof}
We prove the result by induction on $k$. The base case of $k=2$ is obvious because of simplicity of $\FindFrom{2}$ and $\FindFromRight{2}$ procedures.
We first prove the correctness of all the algorithms, assuming there are
no errors. At the end we  explain how to deal with the errors.

\textbf{We start with $\FindFrom{k}$:}
there are different cases to be considered when searching for a $+k$-substring $x[i,j]$ of length $\leq d$.
\begin{enumerate}
    \item Assume that there are $j_1$ and $i_2$ such that $i<j_1<i_2<j$, $|f(x[i,j_1])|=|f(x[i_2,j])|=k-1$ and $\sign(f(x[i,j_1]))=\sign(f(x[i_2,j]))\in s$.
        If $t\in\{i_2,\dots,j\}$, then the algorithm finds $x[i_2,j]$ in Step $1$ and the first invocation of $\FindFirst{k-1}$ in Step $3$  finds $x[i,j_1]$.
        If $t\in\{i,\dots,j_1\}$, then the algorithm finds $x[i,j_1]$ in Step $1$ and the  second invocation of $\FindFirst{k-1}$ in Step $5$ finds $x[i_2,j]$.
        If $j_1<t<i_2$, then the third invocation of $\FindFirst{k-1}$ in Step $6$  finds $x[i_2,j]$ and the  forth invocation of $\FindFirst{k-1}$ in Step $7$  finds $x[i,j_1]$.
    \item Assume that there are $j_1$ and $i_2$ such that $i<i_2<j_1<j$, $|f(x[i,j_1])|=|f(x[i_2,j])|=k-1$ and $\sign(f(x[i,j_1]))=\sign(f(x[i_2,j]))\in s$.
       If $t\in\{i,\dots,j_1\}$, then the algorithm finds $x[i,j_1]$ in Step $1$. After that, it finds $x[i_2,j]$ in Step $4$.
       If $t\in\{j_1+1,\dots,j\}$, then the algorithm finds $x[i_2,j]$ in Step $1$. After that, it finds $x[i,j_1]$ in Step $2$.
\end{enumerate}
By induction, the running time of each $\FindFrom{k-1}$ invocation is $O(\sqrt{d}(\log(r-l))^{0.5(k-3)})$, and the running time of each $\FindFirst{k-1}$ invocation is $O(\sqrt{d}(\log(r-l))^{0.5(k-2)})$.

\textbf{We now look at $\FindFixedLen{k}$:} by
construction and definition of $\FindFrom{k}$, if the algorithm
returns a value, it is a valid substring (with high probability).
If there exists a substring of length at least $d/2$, then any
query to $\FindFrom{k}$ with a value of $t$ in this interval will 
succeed, hence there are at least $d/2$ solutions. Therefore,
by Lemma~\ref{lem:modified_hoyer03}, the algorithm will find one
with high probability and make $O\left(\sqrt{\tfrac{r-l}{d/2}}\right)$ queries. 
Each query has complexity $O(\sqrt{d}(\log(r-l))^{0.5(k-2)})$
by the previous paragraph, hence the running time is bounded by
$O(\sqrt{r-l}(\log(r-l)^{0.5(k-2)})$.

\textbf{We can now analyze $\FindAny{k}$:}  Assume that the shortest $\pm k$-substring $x[i,j]$ is of length $g=j-i+1$. Therefore, there is a $d$ such that $d\leq g \leq 2d$ and the $\FindFixedLen{k}$ procedure  returns a substring for this $d$ with constant success probability. So, the success probability of the randomized algorithm is at least $O(1/\log (l-r))$. Therefore, the amplitude amplification does $O(\sqrt{\log(r-l)})$ iterations. 
The running time of $\FindFixedLen{k}$ is $O(\sqrt{r-l}(\log (r-l))^{0.5(k-2)})$ 
by induction, hence
the total running time is 
$O(\sqrt{r-l}(\log(r-l))^{0.5(k-2)}\sqrt{\log (l-r)})=O(\sqrt{r-l}(\log (r-l))^{0.5(k-1)})$.

\textbf{Finally, we analyze $\FindFirst{k}$:}
See Appendix~\ref{apx:findfirst}.

\textbf{We now turn to error analysis.}
The case of $\FindFrom{k}$ is easy:
the algorithm makes at most $5$ recursive calls, each having a success
probability of $1-\varepsilon$. Hence it will succeed with probability
$(1-\varepsilon)^5$. We can boost this probability to $1-\varepsilon$
by repeating this algorithm a constant number of times. Note that this constant depends on $\varepsilon$.

The analysis of $\FindFixedLen{k}$ follows from \cite{Hoyer03}
and Lemma~\ref{lem:modified_hoyer03}:
since $\FindFrom{k}$ has two-sided error $\varepsilon$, there exists
a search algorithm with two-sided error $\varepsilon$.
\end{proof}


\subsection{The Algorithm for \texorpdfstring{$\dyck_{k,n}$}{DYCK\_\{n,k\}}}\label{sec:dyck-algo}

To solve $\dyck_{k,n}$, we modify the input $x$. As the new input we use $x'=1^k x  0^k$. $\dyck_{k,n}(x)=1$ iff there are no $\pm(k+1)$-substrings in $x'$. This idea is presented in Algorithm \ref{alg:dycknh}. 
\begin{algorithm}[ht]
    \caption{$\textsc{Dyck}_{k,n}(x)$. The Quantum Algorithm for $\dyck_{k,n}$.\label{alg:dycknh}}
    \begin{algorithmic}
        \State $x\gets 1^k  x  0^k$
        \State $v=\FindAny{(k+1)}(0,n+2k-1, \{+1,-1\})$
        \State \Return{$v==\NULL$}
    \end{algorithmic}
\end{algorithm}

\begin{theorem}[Appendix~\ref{sec:Dyck-n,k}]\label{th:upper_bound} Algorithm \ref{alg:dycknh} solves $\dyck_{k,n}$  and
the running time of Algorithm \ref{alg:dycknh} is $O(\sqrt{n}(\log n)^{0.5k})$. The algorithm has two-side error probability  $\varepsilon<0.5$.
\end{theorem}

\section{Lower bounds for Dyck languages}
\label{s:dyck}

\begin{theorem}
\label{t:dycklowerbound-power}
There exist constants $c_{1},c_{2}>0$ such that $Q\left(\dyck_{c_{1}\ell m,c_{2}\left(2m\right)^{\ell}}\right)=\Omega\left(m^{\ell}\right)$.
\end{theorem}

\begin{proof}
We will use the partial Boolean function $\ex{m}{a}{b}=\begin{cases}
1, & \text{if \ensuremath{\left|x\right|_{0}=a}}\\
0, & \text{if \ensuremath{\left|x\right|_{0}=b.} }
\end{cases}$ 

We prove the theorem by a reduction $\left(\exnn\right) ^{\ell}\leqslant\dyck_{c_{1}\ell m,c_{2}\left(2m\right)^{\ell}}$, with
the reduction described in Appendix~\ref{sec:reduction}.
It is known that $Adv^\pm\left(\exnn\right) \geq Adv\left(\exnn\right)>m$ \cite[Theorem~5.4]{ambainis2002quantum}.
The Adversary bound composes even for partial Boolean functions \cite[Lemma~1]{kimmel2012quantum}, therefore $Q\left( \left(\exnn\right)^{\ell}\right)=\Omega\left(m^{\ell}\right)$.
Via the reduction the same bound applies to $\dyck_{c_{1}\ell m,c_{2}\left(2m\right)^{\ell}}$.
\end{proof}

\begin{theorem}
\label{t:dycklowerbound}
For any $\epsilon > 0$, there exists $c>0$ such that $Q\left(\dyck_{c\log n,n}\right)=\Omega\left(n^{1-\epsilon}\right).$
\end{theorem}

\begin{proof}
For any $\epsilon>0$, there exists an $m$ such that $Adv^\pm\left(\ex{2m}{m}{m+1}\right)\geq\left(2m\right)^{1-\epsilon}$.
Without loss of generality we may assume that $(2m)^\ell=n$. From Theorem \ref{t:dycklowerbound-power} with $\ell=\log_{2m}n$ we obtain  $c_2\left(2m\right)^{\ell}=c_2 n$
and height $c_1 m\ell=\Theta\left(\log n\right)$. The query complexity
is at least $\left(\left(2m\right)^{1-\epsilon}\right)^{\ell}=\left(\left(2m\right)^{\ell}\right)^{1-\epsilon}=n^{1-\epsilon}$.
Therefore $Q\left(\dyck_{c \log n,n}\right)=\Omega\left(n^{1-\epsilon}\right)$.
\end{proof}

For constant depths the following bound can be derived:
\begin{theorem}
There exists a constant $c_1>0$ such that $Q(\dyck_{c_1\ell,n})=\Omega(2^{\frac{\ell}{2}}\sqrt{n}).$
\end{theorem}
\begin{proof}
    Let $m=4$ in the Theorem \ref{t:dycklowerbound-power}. Then, $Q\left(\dyck_{c_{1}\ell,c_{2}8^{\ell}}\right)=\Omega\left(4^{\ell}\right)$ for some constants $c_{1},c_{2}>0$. Consider the function $\andf_\frac{n}{c_{2}8^{\ell}} \circ \dyck_{c_{1}\ell,c_{2}8^{\ell}}$ with a promise that $\andf_k$ has as an input either $k$ or $k-1$ ones. Then,
    \[ Q\left({\andf}_\frac{n}{c_{2}8^{\ell}} \circ \dyck_{c_{1}\ell,c_{2}8^{\ell}}\right) = \Theta\left(Adv^{\pm}\left({\andf}_\frac{n}{c_{2}8^{\ell}} \circ \dyck_{c_{1}\ell,c_{2}8^{\ell}}\right)\right) \mbox{~and} \]
    \[ Adv^{\pm}\left({\andf}_\frac{n}{c_{2}8^{\ell}} \circ \dyck_{c_{1}\ell,c_{2}8^{\ell}}\right) 
    \geq 
    Adv^{\pm}\left({\andf}_\frac{n}{c_{2}8^{\ell}}\right)  Adv^{\pm}\left(\dyck_{c_{1}\ell,c_{2}8^{\ell}}\right)
    =\Omega\left(2^\frac{\ell}{2}\sqrt{n}\right) ,\]
   with the second step following from the composition of $Adv^{\pm}$ for partial functions \cite{kimmel2012quantum}.
    This implies the same lower bound on $\dyck_{c_{1}\ell,n}$ because 
    the computation of the composition $\andf_\frac{n}{c_{2}8^{\ell}} \circ \dyck_{c_{1}\ell,c_{2}8^{\ell}}$ can be straightforwardly reduced to $\dyck_{c_{1}\ell,n}$ by a simple concatenation of $\dyck_{c_{1}\ell,c_{2}8^{\ell}}$ instances.
\end{proof}

\section{Quantum complexity of  \textsc{st-Connectivity} in grids}
\label{s:connectivity}

\subsection{Quantum complexity of \texorpdfstring{$\dirtwod_{n, k}$}{2D-DConnectivity}}
\label{ss:lbdirected}

\begin{theorem}
\label{t:dirconlowerbound}
For any $n\geq k$ and $\epsilon>0$,
$Q(\dirtwod_{n, k})=\Omega\left((\sqrt{n}k)^{1-\epsilon}\right)$.
\end{theorem}
In particular, if we have a square grid then
\begin{cor}
For any $\epsilon > 0$, $Q(\dirtwod_{n, n})=\Omega\left(n^{1.5-\epsilon}\right)$.
\end{cor}

\begin{proof}[Proof of Theorem \ref{t:dirconlowerbound}]
For any sequence $w$ of $m$ opening and closing parentheses it is possible to plot the changes of depth, i.e., the number of opening parentheses minus the number of closing parentheses, for all prefixes of the sequence, see Figure \ref{f:pyramids}.
\begin{figure}[H]
    \centering
    \begin{tikzpicture}[scale=0.4]
    \tikzmath{
        int \x,\y,\newx,\newy,\depth;
        \y=0; \newy=0;
        \x=0; \newx=0;
        \depth=4;
        for \delta in {1,1,-1,-1,1,1,1,-1,-1,1,-1,-1}%
        {
            \newx=\x+1;
            if \delta > 0 then %
            {
                let \paren = (;
            } else 
            {
                let \paren = );
            };
            \newy=\y+\delta;
            {
                \path (\x,0) -- (\newx,0) node [midway, below,yshift=-20pt] {\LARGE \texttt{\paren}};
                \draw[-Stealth,very thick] (\x,\y) -- (\newx,\newy);
            };
            \y=\newy;
            \x=\newx;
        };
        {
            \draw[->] (-0.5,0) -- node[above,pos=1] {$x$} (\newx+0.5, 0);
            \draw[->] (0,-0.5) -- node[left,pos=1] {$y$} (0, \depth+0.5);
            \draw[dashed] (-0.5,\depth) -- node[above, pos=0.9] {$y=d=4$} (\newx+0.5, \depth);
            \draw (-0.5, 0.3) node {$(0,0)$};
            \path [preaction={contour=8pt,rounded corners=10,draw}] (0,0) -- (\depth,\depth) -- (\newx-\depth, \depth) -- (\newx, 0) -- cycle;
        };    
    };
    \end{tikzpicture}    
    \caption{Representation of the Dyck word ``\texttt{(())((())())}''}
    \label{f:pyramids}
\end{figure}
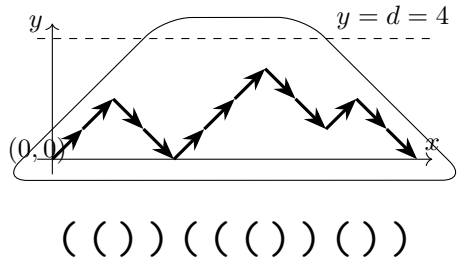
We can connect neighboring points by vectors $(1,1)$ and $(1,-1)$ corresponding to opening and closing parentheses respectively. Clearly $w\in L_d$ if and only if the path starting at the origin $(0,0)$ ends at $(m,0)$ and never crosses $y=0$ and $y=d$. Consequently a path corresponding to $w\in L_d$ always remains within the trapezoid bounded by $y=0$, $y=d$, $y=x$, $y=-x+m$. This suggests a way of mapping $\dyck_{d,m}$ to the $\dirtwod_{n,k}$ problem:

\noindent\begin{minipage}{0.58\linewidth}
\begin{enumerate}
    \item
    An opening parenthesis in position $i$ corresponds to a ``column'' of upwards sloping available edges $(i-1,l)\rightarrow (i,l+1)$ for all $l \in \{0,1, \ldots, d-1\}$ such that $i-1+l$ is even.
    A closing parenthesis in position $i$ corresponds to downwards sloping available edges $(i-1,l)\rightarrow (i,l-1)$ for all $l \in \{1, \ldots, d\}$ such that $i-1+l$ is even. See Figure \ref{f:varmap}.
    \item The edges outside the trapezoid adjacent to the trapezoid are forbidden (see Figure \ref{f:pyramids2}), i.e., it is sufficient to ``insulate'' the trapezoid by a single layer of forbidden edges. The only exception are the edges adjacent to the $(0,0)$ and $(m,0)$ vertex as those will be used in the construction (step 4).
\end{enumerate}
\end{minipage}%
\hfill%
\begin{minipage}{0.38\textwidth}
    \begin{figure}[H]
    \centering
    \begin{tikzpicture}[scale=0.75]
    \path (0,0) -- (0,6) node [midway, xshift=-20pt] {\LARGE \texttt{(}$\implies$};
    \path (4,0) -- (4,6) node [midway, xshift=-20pt] {\LARGE \texttt{)}$\implies$};
        \tikzmath{
            for \y in {0,...,2}%
            {
                {
                    \draw[-Stealth] (0,2*\y) -- ++(1,1);
                    \draw[-Stealth,dotted] (0,2*\y+2) -- ++(1,-1);
                    \draw[-Stealth,dotted] (4,2*\y) -- ++(1,1);
                    \draw[-Stealth] (4,2*\y+2) -- ++(1,-1);
                };
            };
        };
    \end{tikzpicture}    
    \caption{Mapping of $\dyck_{d,m}$ variables to \dirtwodtext{}}
    \label{f:varmap}
    \end{figure}
\end{minipage}
\begin{figure}[H]
    \centering
    \begin{tikzpicture}[scale=0.75]
    \tikzmath{
    int \x,\y,\newx,\newy,\depth;
    \y=0;
    \x=0;
    \depth=4;
    for \delta in {1,1,-1,-1,1,1,1,-1,-1,1,-1,-1}%
    {
        \newx=\x+1;
        if (\delta > 0) then
        {
            \paren = "(";
        }   
        else
        {
            \paren = ")";
        };
        \newy=\y+\delta;
        {
            \path (\x,0) -- (\newx,0) node [midway, below,yshift=-20pt] {\LARGE \texttt{\paren}};
            \draw[-Stealth,very thick] (\x,\y) -- (\newx,\newy);
        };
        \y=\newy;
        \x=\newx;
    };
    {
        \draw[-Stealth,dotted] (0,0) -- (1,-1);
        \draw[-Stealth,dotted] (1,-1) -- (2,0);
        \draw[-Stealth,dotted] (2,0) -- (3,-1);
        \draw[-Stealth,dotted] (3,-1) -- (4,0);
        \draw[-Stealth,dotted] (4,0) -- (5,-1);
        \draw[-Stealth,dotted] (5,-1) -- (6,0);
        \draw[-Stealth,dotted] (6,0) -- (7,-1);
        \draw[-Stealth,dotted] (7,-1) -- (8,0);
        \draw[-Stealth,dotted] (8,0) -- (9,-1);
        \draw[-Stealth,dotted] (9,-1) -- (10,0);
        \draw[-Stealth,dotted] (10,0) -- (11,-1);
        \draw[-Stealth,dotted] (11,-1) -- (12,0);

        \draw[-Stealth,dotted] (0,2) -- (1,1);
        \draw[-Stealth,dotted] (1,3) -- (2,2);
        \draw[-Stealth,dotted] (2,4) -- (3,3);
        \draw[-Stealth,dotted] (3,5) -- (4,4);
        \draw[-Stealth,dotted] (4,4) -- (5,5);
        \draw[-Stealth,dotted] (5,5) -- (6,4);
        \draw[-Stealth,dotted] (6,4) -- (7,5);
        \draw[-Stealth,dotted] (7,5) -- (8,4);
        \draw[-Stealth,dotted] (8,4) -- (9,5);
        \draw[-Stealth,dotted] (9,3) -- (10,4);
        \draw[-Stealth,dotted] (10,2) -- (11,3);
        \draw[-Stealth,dotted] (11,1) -- (12,2);
        
        \draw[-Stealth] (3,3) -- (4,2);
        \draw[-Stealth] (4,2) -- (5,3);
        \draw[-Stealth] (5,3) -- (6,4);
        \draw[-Stealth] (6,0) -- (7,1);
        \draw[-Stealth] (7,1) -- (8,0);
        \draw[-Stealth] (8,4) -- (9,3);

        \draw[-Stealth,dotted] (1,1) -- (2,0);
        \draw[-Stealth,dotted] (2,0) -- (3,1);
        \draw[-Stealth,dotted] (2,2) -- (3,3);
        \draw[-Stealth,dotted] (3,1) -- (4,2);
        \draw[-Stealth,dotted] (3,3) -- (4,4);
        \draw[-Stealth,dotted] (4,2) -- (5,1);
        \draw[-Stealth,dotted] (4,4) -- (5,3);
        \draw[-Stealth,dotted] (5,1) -- (6,0);
        \draw[-Stealth,dotted] (5,3) -- (6,2);
        \draw[-Stealth,dotted] (6,2) -- (7,1);
        \draw[-Stealth,dotted] (6,4) -- (7,3);
        \draw[-Stealth,dotted] (7,1) -- (8,2);
        \draw[-Stealth,dotted] (7,3) -- (8,4);
        \draw[-Stealth,dotted] (8,0) -- (9,1);
        \draw[-Stealth,dotted] (8,2) -- (9,3);
        \draw[-Stealth,dotted] (9,1) -- (10,0);
        \draw[-Stealth,dotted] (9,3) -- (10,2);
        \draw[-Stealth,dotted] (10,0) -- (11,1);

        \draw[->] (-0.5,0) -- (\newx+0.5, 0);
        \draw[->] (0,-0.5) -- (0, \depth+0.5);
        \draw[dashed] (-0.5,\depth) -- (\newx+0.5, \depth);
        \path [preaction={contour=8pt,rounded corners=10,draw,dashed}] (0,0) -- (\depth,\depth) -- (\newx-\depth, \depth) -- (\newx, 0) -- cycle;
        \draw[-Stealth] (13,3.5) -- node[pos=1,right,align=left] {Available edges} ++(1,0);
        \draw[-Stealth,very thick] (13,2.5) -- node[pos=1,right,align=left] {Available edges}
        node[pos=1,right,align=left,yshift=-2.5ex,xshift=-5ex] {reachable from origin} ++(1,0);
        \draw[-Stealth,dotted] (13,1.1) -- node[pos=1,right,align=left] {Forbidden edges} ++(1,0);
    };
    }; 
    \end{tikzpicture}    
    \caption{Mapping of a complete input corresponding to Dyck word ``\texttt{(())((())())}'' to $\dirtwod$}
    \label{f:pyramids2}
\end{figure}
\begin{enumerate}
    \setcounter{enumi}{2}
    \item Rotate the trapezoid by 45 degrees counterclockwise. This isolated trapezoid can be embedded in a directed grid and its starting and ending vertices are connected by a path if and only if the corresponding input word is valid. 
    \item Finally we can lay multiple independent trapezoids side by side and connect them in parallel forming an $\orf_t$ of $\dyck_{d,m}$ instances; see Figure \ref{f:trapezoids}.
\end{enumerate}



\begin{figure}[H]
    \centering
    \begin{minipage}{0.5\textwidth}
        \centering
        \begin{tikzpicture}[scale=0.33]
            \draw [step=1,dotted] (0,0) grid (20,10);
    
            \foreach \x in {0,5,...,10} 
            {
            	\draw[thick,fill=gray!30,dashed] (\x,0)--++(0,4)--++(5,5)--++(4,0) coordinate (exit\x) -- cycle;
            	\draw[very thick,->] (exit\x)--++(0, 1);
            }
    
            \draw[very thick,->] (0,0)--(10,0);
            \draw[very thick,->] (9,10)--(20,10);
        \end{tikzpicture}    
        \captionof{figure}{Reduction\\ $\orf_t\circ \dyck \leqslant \dirtwod$}
        \label{f:trapezoids}
    \end{minipage}%
    \begin{minipage}{0.5\textwidth}
        \centering
      
        \resizebox{\textwidth}{!}{%
        \begin{tikzpicture}[scale=0.25]
            \draw [dotted] (0,0) grid (44,20);
    
            	\draw[thick,fill=gray!30,dashed] (0,0)--++(0,4)--++(16,16)--++(2,-2) -- cycle;
            	\draw[thick,fill=gray!30,dashed] (19,17)--++(2,-2)--++(-12,-12)--++(-2,2) -- cycle;
            	\draw[thick,fill=gray!30,dashed] (12,0)--++(-2,2)--++(18,18)--++(2,-2) -- cycle;
            	\draw[thick,fill=gray!30,dashed] (31,17)--++(2,-2)--++(-12,-12)--++(-2,2) -- cycle;
            	\draw[thick,fill=gray!30,dashed] (24,0)--++(-2,2)--++(18,18)--++(4,0) -- cycle;
                
                \draw[thick] (16,20)--++(5,0)--++(0,-5);
                \draw[thick] (17,19)--++(3,0)--++(0,-3);
                \draw[thick] (18,18)--++(1,0)--++(0,-1);
                \draw[thick] (28,20)--++(5,0)--++(0,-5);
                \draw[thick] (29,19)--++(3,0)--++(0,-3);
                \draw[thick] (30,18)--++(1,0)--++(0,-1);
                \draw[thick] (12,0)--++(-5,0)--++(0,5);
                \draw[thick] (11,1)--++(-3,0)--++(0,3);
                \draw[thick] (10,2)--++(-1,0)--++(0,1);
                \draw[thick] (24,0)--++(-5,0)--++(0,5);
                \draw[thick] (23,1)--++(-3,0)--++(0,3);
                \draw[thick] (22,2)--++(-1,0)--++(0,1);
    
        \end{tikzpicture}    
        }
        \captionof{figure}{Folding of a long $\dyck$ instance in an undirected grid}
        \label{f:foldgadget}
    \end{minipage}
\end{figure}

This concludes the reduction $\orf_t\circ \dyck_{d,m} \leqslant \dirtwod_{n,k}$, where $n=(d+1)(t-1)+\frac{m}{2}+1$ and $k=\frac{m}{2}+1$. By the well known composition result of Reichardt \cite{Reichardt11} we know that $Q(\orf_t\circ \dyck_{d,m})=\Theta(Q(\orf_t)\cdot Q(\dyck_{d,m}))$.
All that remains is to pick suitable $t$, $d$ and $m$ for the proof to be complete. Let $k$ be the vertical dimension of the grid and $k\leq n$. Then we take $m=\Theta(k)$, $d=\log{m}$ and $t=\frac{n}{d}$.
\end{proof}

Constructing a non-trivial quantum algorithm appears to be difficult and 
we conjecture that the actual complexity may be $\Omega(nk)$, except for 
the case when $k$ is small, compared to $n$. For very small $k$ (up to 
$k=\Theta(\frac{\log n}{\log \log n})$), a better quantum algorithm is possible.

\begin{theorem}[Appendix~\ref{sec:narrowgrid}]\label{th:upper_dirtwod}
$Q(\dirtwod_{n,k})=O\left(\sqrt{n}\log_2^{k/2} n\right)$. Moreover, there is a time-efficient quantum query algorithm that solves $\dirtwod_{n,k}$ in time $O\left(\sqrt{n}\log_2^{k/2+O(1)} n\right)$.
\end{theorem}


\subsection{Lower bounds for \texorpdfstring{$\undirtwod_{n,k}$}{2D-Connectivity}}
Even though it is possible to use the construction from Section \ref{ss:lbdirected} to give a lower bound of $\Omega\left((\sqrt{n}k)^{1-\epsilon}\right)$ for the undirected case because the paths for each instance of $\dyck$ never bifurcate or merge, this lower bound can be further improved to a nearly tight estimate.

\begin{theorem}
\label{t:undirconlowerbound}
For any $n\geq k$, $k=\Omega(\log{n})$, $\epsilon>0$,
$Q(\undirtwod_{n,k})=\Omega\left((nk)^{1-\epsilon}\right)$.
\end{theorem}

\begin{proof}
    We start off by representing an input as a path in a trapezoid, see Figure \ref{f:pyramids2}. But now instead of connecting multiple instances of $\dyck$ in parallel we will embed one long instance by folding it when it hits the boundary of the graph. To implement a fold we will use simple gadgets depicted in Figure \ref{f:foldgadget}.

            

This way a $\dyck$ instance of length $m$ and depth $\log{m}$ can be embedded in an $n \times k$ grid such that $\frac{nk}{\log{m}}=\Theta(m)$. Using Theorem \ref{t:dycklowerbound} we conclude that solving \undirtwodtext${}_{n,k}$ requires at least $\Omega\left((nk)^{1-\epsilon}\right)$ quantum queries.
\end{proof}

\subsection{Lower bounds for \texorpdfstring{$d$}{d}-dimensional grids}
For undirected $d$-dimensional grids we give a tight bound on the number of queries required to solve connectivity.
\begin{theorem}
    For any $\epsilon>0$, for undirected $d$-dimensional grids of size $n_1\times n_2 \times \ldots \times n_d$ that are not ``almost-one-dimensional'', i.e., there exists $i\in [d]$ such that $\frac{\prod_{j=1}^d{n_j}}{n_i}=\Omega(\log{n_i})$:
    \[Q(\undirdd_{n_1, n_2,\ldots, n_d}) = \Omega((n_1 \cdot n_2 \cdot \ldots \cdot n_d)^{1-\epsilon}).\]
\end{theorem}
\begin{proof}
For the purposes of this theorem, it is more convenient to refer to $n_1 \times \ldots \times n_d$ sized grids as $n_1'\times \ldots \times n_d'$ sized where $n_i'=n_i+1$. Then the theorem follows from the 2D case by iteratively using the fact that a $d$-dimensional grid of size $n_1'\times n_2'\times \ldots \times n_{d-1}'\times n_d'$ contains as a subgraph a $(d-1)$-dimensional grid of size $n_1' \times n_2' \times \ldots \times n_{d-2}' \times n_{d-1}'n_d'$. One way to see this is to consider a bijective mapping of the vertices $(x_1, \ldots, x_{d-1}, x_d)$ to $(x_1, \ldots, x_{d-2},x_d n_{d-1}'+x_{d-1})$ if $x_d$ is even and to $(x_1, \ldots, x_{d-2},x_d n_{d-1}'+n_{d-1}'-1-x_{d-1})$ if $x_d$ is odd. It is a bijection because $x_d$ and $x_{d-1}$ can be recovered from $x_dn_{d-1}'+n_{d-1}'-1-x_{d-1}$ by computing the quotient and remainder on division by $n_{d-1}'$. One can view this procedure as ``folding'' where we take layers (vertices corresponding to some $x_d=l$) and fold them into the $(d-1)$-st dimension alternating the direction of the layers depending on the parity of the layer $l$. For this procedure to place the starting and ending vertices the furthest apart, it requires that $n_d'$ is an odd number. Otherwise we embed a smaller subgraph $n_1'\times \ldots \times n_{d-1}'\times (n_d'-1)$ and add an edge $(n_1, \ldots, n_{d-1}, n_{d}-1)$ to $(n_1, \ldots, n_{d-1}, n_{d})$. In the end we obtain a lower bound of $\Omega(((((((n_d'-1)n_{d-1}'-1)n_{d-2}'-1)\ldots )n_2'-1)n_1')^{1-\epsilon})=\Omega((n_1 \cdot n_2 \cdot \ldots \cdot n_d)^{1-\epsilon})$.
\end{proof}

\noindent For directed $d$-dimensional grids we can only slightly improve over the $n^{\frac{d}{2}}$ trivial lower bound.
\begin{theorem}\label{t:dirdlowerbound}
    For directed $d$-dimensional grids of size $n_1\times n_2 \times \ldots \times n_d$ such that $n_1\leq n_2 \leq \ldots \leq n_d$ and $\epsilon>0$,
    $Q(\dirdd_{n_1,n_2,\ldots,n_d}) = \Omega((n_{d-1}\prod_{i=1}^d{n_i})^{\frac{1}{2}-\epsilon})$.
\end{theorem}

\begin{cor}
For directed $d$-dimensional grids of size $n\times n \times \ldots \times n$ and $\epsilon>0$,\\
$Q(\dirdd_{n,n,\ldots,n}) = \Omega(n^{\frac{d+1}{2}-\epsilon})$.
\end{cor}
\begin{proof}[Proof of Theorem \ref{t:dirdlowerbound}]
    For each $I\in \{0,1,\ldots,n_1\}\times \{0,1,\ldots,n_1\}\times \ldots \times \{0,1,\ldots,n_{d-2}\}$ we  take take a $2$-dimensional hard instance $G_I$ of $\dirtwod_{n_{d-1},n_d}$ having query complexity $\Omega(n_{d-1}^{1-\epsilon}n_d^{
    \frac{1}{2}-\epsilon})$. We then connect them in parallel like so:
    \begin{itemize}
        \item Include the entire $(d-2)$-dimensional subgrid from $(0, \ldots, 0)$ to $(n_1, n_2, \ldots, n_{d-2}, 0, 0)$ and similarly the subgrid from $(0, 0, \ldots, 0, n_{d-1}, n_d)$ to $(n_1, n_2, \ldots, n_{d-2}, n_{d-1}, n_d)$;
        \item For each $I\in \{0,1,\ldots,n_1\}\times \{0,1,\ldots,n_1\}\times \ldots \times \{0,1,\ldots,n_{d-2}\}$ embed the instance $G_I$ in the subgrid $(I, 0, 0)$ to $(I, n_{d-1}, n_d)$;
        \item Forbid all other edges.
    \end{itemize}
    This construction computes $\orf_{\prod_{i=1}^{d-2}{(n_i+1)}}\circ \dirtwod_{n_{d-1},n_d}$ whose complexity is at least $\Omega(\sqrt{\prod_{i=1}^{d-2}{(n_i+1)}}n_{d-1}^{1-\epsilon}n_d^{\frac{1}{2}-\epsilon})=\Omega((n_{d-1}\prod_{i=1}^d{n_i})^{\frac{1}{2}-\epsilon})$.
\end{proof}

\section{Directions for future works}
\label{s:concl}

Some directions for future work are:
\begin{enumerate}
    \item {\bf Better algorithm/lower bound for the directed 2D grid?} 
    Can we find an $o(n^2)$ query quantum algorithm or improve our lower bound? A nontrivial quantum algorithm would be particularly interesting, as it may imply a quantum algorithm for edit distance.  
    \item {\bf Quantum algorithms for directed connectivity?} More generally, can we come up with better quantum algorithms for directed connectivity? The span program method used by Belovs and Reichardt \cite{BelovsR12} for the undirected connectivity does not work in the directed case. As a result, the quantum algorithms for directed connectivity are typically based on Grover's search in various forms, from simply speeding up depth-first/breadth-first search to more sophisticated approaches \cite{A+19}. Developing other methods for directed connectivity would be very interesting.
    \item {\bf Quantum speedups for dynamic programming.} Dynamic programming is a widely used algorithmic method for classical algorithms and it would be very interesting to speed it up quantumly. This has been the motivating question for both the connectivity problem on the directed 2D grid studied in this paper and a similar problem for the Boolean hypercube in \cite{A+19} motivated by algoritms for Travelling Salesman Problem. There are many more dynamic programming algorithms and exploring quantum speedups of them would be quite interesting.  
\end{enumerate}

\section*{Acknowledgements}
The authors would like to thank the anonymous reviewers for their constructive comments and suggestions.

\phantomsection
\addcontentsline{toc}{chapter}{References}
\bibliography{references}

\begin{thebibliography}{10}

\bibitem{Aaronson2018AQQ}
Scott Aaronson, Daniel Grier, and Luke Schaeffer.
\newblock A quantum query complexity trichotomy for regular languages.
\newblock {\em Electronic Colloquium on Computational Complexity (ECCC)},
  26:61, 2018.

\bibitem{ambainis2002quantum}
Andris Ambainis.
\newblock Quantum lower bounds by quantum arguments.
\newblock {\em Journal of Computer and System Sciences}, 64(4):750--767, 2002.

\bibitem{A+19}
Andris Ambainis, Kaspars Balodis, Janis Iraids, Martins Kokainis, Krisjanis
  Prusis, and Jevgenijs Vihrovs.
\newblock Quantum speedups for exponential-time dynamic programming algorithms.
\newblock In {\em Proceedings of the Thirtieth Annual {ACM-SIAM} Symposium on
  Discrete Algorithms, {SODA} 2019, San Diego, California, USA, January 6-9,
  2019}, pages 1783--1793, 2019.
\newblock URL: \url{https://doi.org/10.1137/1.9781611975482.107}, \href
  {http://dx.doi.org/10.1137/1.9781611975482.107}
  {\path{doi:10.1137/1.9781611975482.107}}.

\bibitem{Ambainis+10}
Andris Ambainis, Andrew~M. Childs, Ben Reichardt, Robert Spalek, and Shengyu
  Zhang.
\newblock Any {AND-OR} formula of size {N} can be evaluated in time
  n\({}^{\mbox{1/2+o(1)}}\) on a quantum computer.
\newblock {\em {SIAM} J. Comput.}, 39(6):2513--2530, 2010.
\newblock URL: \url{https://doi.org/10.1137/080712167}, \href
  {http://dx.doi.org/10.1137/080712167} {\path{doi:10.1137/080712167}}.

\bibitem{backurs2015edit}
Arturs Backurs and Piotr Indyk.
\newblock Edit distance cannot be computed in strongly subquadratic time
  (unless {SETH} is false).
\newblock In {\em Proceedings of the forty-seventh annual ACM symposium on
  Theory of computing}, pages 51--58. ACM, 2015.

\bibitem{BelovsR12}
Aleksandrs Belovs and Ben~W. Reichardt.
\newblock Span programs and quantum algorithms for st-connectivity and claw
  detection.
\newblock In {\em Algorithms - {ESA} 2012 - 20th Annual European Symposium,
  Ljubljana, Slovenia, September 10-12, 2012. Proceedings}, pages 193--204,
  2012.
\newblock URL: \url{https://doi.org/10.1007/978-3-642-33090-2\_18}, \href
  {http://dx.doi.org/10.1007/978-3-642-33090-2\_18}
  {\path{doi:10.1007/978-3-642-33090-2\_18}}.

\bibitem{boroujeni2018approximating}
Mahdi Boroujeni, Soheil Ehsani, Mohammad Ghodsi, MohammadTaghi HajiAghayi, and
  Saeed Seddighin.
\newblock Approximating edit distance in truly subquadratic time: Quantum and
  {MapReduce}.
\newblock In {\em Proceedings of the Twenty-Ninth Annual ACM-SIAM Symposium on
  Discrete Algorithms}, pages 1170--1189. SIAM, 2018.

\bibitem{buhrman2019quantum}
Harry Buhrman, Subhasree Patro, and Florian Speelman.
\newblock The quantum strong exponential-time hypothesis, 2019.
\newblock \href {http://arxiv.org/abs/1911.05686} {\path{arXiv:1911.05686}}.

\bibitem{C+18}
Diptarka Chakraborty, Debarati Das, Elazar Goldenberg, Michal Kouck{\'{y}}, and
  Michael~E. Saks.
\newblock Approximating edit distance within constant factor in truly
  sub-quadratic time.
\newblock In {\em 59th Annual IEEE Symposium on Foundations of Computer Science
  (FOCS), Paris, France, Oct 7-9, 2018}, pages 979--990, 2018.
\newblock \href {http://arxiv.org/abs/1810.03664} {\path{arXiv:1810.03664}}.

\bibitem{Hoyer03}
Peter H{\o}yer, Michele Mosca, and Ronald de~Wolf.
\newblock Quantum search on bounded-error inputs.
\newblock In Jos C.~M. Baeten, Jan~Karel Lenstra, Joachim Parrow, and
  Gerhard~J. Woeginger, editors, {\em Automata, Languages and Programming},
  pages 291--299, Berlin, Heidelberg, 2003. Springer Berlin Heidelberg.

\bibitem{kimmel2012quantum}
Shelby Kimmel.
\newblock Quantum adversary (upper) bound.
\newblock In {\em International Colloquium on Automata, Languages, and
  Programming}, pages 557--568. Springer, 2012.

\bibitem{k2014}
Robin Kothari.
\newblock An optimal quantum algorithm for the oracle identification problem.
\newblock In {\em 31st International Symposium on Theoretical Aspects of
  Computer Science}, page 482, 2014.

\bibitem{ll2016}
C.~Y.-Y. Lin and H.-H. Lin.
\newblock Upper bounds on quantum query complexity inspired by the
  elitzur--vaidman bomb tester.
\newblock {\em Theory of Computing}, 12(18):1--35, 2016.

\bibitem{Reichardt11}
Ben~W. Reichardt.
\newblock Reflections for quantum query algorithms.
\newblock In {\em Proceedings of the Twenty-second Annual ACM-SIAM Symposium on
  Discrete Algorithms}, SODA '11, pages 560--569, Philadelphia, PA, USA, 2011.
  Society for Industrial and Applied Mathematics.
\newblock URL: \url{http://dl.acm.org/citation.cfm?id=2133036.2133080}.

\bibitem{Reichardt14}
Ben~W. Reichardt.
\newblock Span programs are equivalent to quantum query algorithms.
\newblock {\em {SIAM} J. Computing}, 43(3):1206--1219, 2014.
\newblock URL: \url{https://doi.org/10.1137/100792640}, \href
  {http://dx.doi.org/10.1137/100792640} {\path{doi:10.1137/100792640}}.

\bibitem{wagner1974string}
Robert~A Wagner and Michael~J Fischer.
\newblock The string-to-string correction problem.
\newblock {\em Journal of the ACM (JACM)}, 21(1):168--173, 1974.

\end{thebibliography}

\bookmarksetupnext{level=0}
\newpage
\appendix
  
  

\section{An Algorithm for the \texorpdfstring{$\FindFrom{k}$}{FindFrom\_k} Subroutine}\label{sec:desc_k_substring_base}

\begin{algorithm}[h]
    \caption{$\FindFrom{k}(l,r,t,d,s)$. \label{alg:substringk_base}}
    \begin{algorithmic}
    \State $v=(i_1,j_1,\sigma_1)\leftarrow \FindFrom{k-1}(l,r,t,d-1,\{+1,-1\})$
    \If{$v \neq \NULL$}\Comment{if $t$ is inside a $\pm(k-1)$-substring}
        \State $v'=(i_2,j_2,\sigma_2)\leftarrow \FindFromRight{k-1}(l,r,i_1-1,d-1,\{+1,-1\})$
        \If{$v'=\NULL$}
            \State $v'=(i_2,j_2,\sigma_2)\gets\FindFirst{k-1}(\min(l,j_1-d+1),i_1-1,\{+1,-1\},left)$
        \EndIf
        \If {$v'\neq\NULL$ and $\sigma_2\neq\sigma_1$}
        \State $v'\gets\NULL$
        \EndIf
        \If{$v'=\NULL$}
            \State $v'=(i_2,j_2,\sigma_2)\leftarrow \FindFrom{k-1}(l,r,j_1+1,d-1,\{+1,-1\})$
        
        \If{$v'=\NULL$}
            \State $v'=(i_2,j_2,\sigma_2)\leftarrow \FindFirst{k-1}(j_1+1,\min(i_1+d-1,r),\{+1,-1\},right)$
        \EndIf
        \EndIf
        \If{$v'=\NULL$}
            \State \Return{$\NULL$}
        \EndIf
    \Else
        \State $v=(i_1,j_1,\sigma_1)\gets\FindFirst{k-1}(t,\min(t+d-1,r),\{+1,-1\},right)$
        \If{$v=\NULL$}
            \State \Return{$\NULL$}
        \EndIf
        \State $v'=(i_2,j_2,\sigma_2)\gets\FindFirst{k-1}(\max(l,t-d+1),t),\{+1,-1\},left)$
        \If{$v'=\NULL$}
            \State \Return{$\NULL$}
        \EndIf
    \EndIf
    \If{$\sigma_1=\sigma_2$ and $\sigma\in s$ and $\max(j_1,j_2)-\min(i_1,i_2)+1\leq d$}
        \State\Return{$(\min(i_1,i_2),\max(j_1,j_2),\sigma_1)$}
    \Else
        \State \Return{$\NULL$}
    \EndIf
    \end{algorithmic}
\end{algorithm}

\section{Proof of Lemma \ref{lem:modified_hoyer03}}
\label{sec:bounded}

    The main loop of the algorithm of \cite{Hoyer03} is the following,
    assuming the algorithms have error at most $1/9$:
    \begin{itemize}
        \item for $m=0$ to $\lceil\log_9 n\rceil$-1 do:
            \begin{enumerate}
                \item run $A_m$ 1000 times,
                \item verify the 1000 measurements, each by $O(\log n)$ runs of the corresponding algorithm,
                \item if a solution has been found, then output a solution and stop
            \end{enumerate}
        \item Output ‘no solutions’
    \end{itemize}
    The key of the analysis is that if the (unknown) number $t$
    of solutions lies in the interval $[n/9^{m+1},n/9^m]$, then
    $A_m$ succeeds with constant probability. In all cases, if there are no solutions, $A_m$ will never succeeds with high probability (ie the algorithm only applies good solutions).
    
    In our case, we allow the algorithm to return anything (including
    $\NULL$) if $t<T$. This means that we only care about the values of
    $m$ such that $n/9^{m}\geqslant T$, that is
    $m\leqslant\log_9\tfrac{n}{T}$. Hence, we simply run the algorithm
    with this new upper bound for $d$ and it will satisfy our
    requirements with constant probability. The complexity is
    
    $
        \sum\limits_{m=0}^{\left\lfloor log_9\tfrac{n}{T}\right\rfloor}1000\cdot O(3^m)+1000\cdot O(\log n)=O(3^{\log_9\tfrac{n}{T}})=O(\sqrt{n/T}).
    $

\section{\texorpdfstring{$\FindFirst{k}$}{FindFirst\_k} Algorithm's Description, Complexity and Proof of Correctness }\label{apx:findfirst}
\subsection{\texorpdfstring{$\FindFixedPos{k}$}{FidFixedPos\_k}}
Let us first describe a subroutine used by $\FindFirst{k}$.

$\FindFixedPos{k}(l,r,t,s,left)$ locates the leftmost substring $x[i,j]$ such that $|f(x[i,j])|=k$ and $\sign(f(x[i,j]))\in s$ , i.e. $i\leq t \leq j$ and there is no $x[i',j']$ such that $i'\leq t\leq j'$, $i'<i$ and $f(x[i',j'])=f(x[i,j])$.

The procedure is similar to $\FindAny{k}$. First, we consider a randomized algorithm that uniformly chooses $d$ as a power of $2$ that is at most $r-l$. For this $d$, it runs $\FindFrom{k}(l,r,t,d,s)$ algorithm and searches for a non-NULL result. The probability of getting a correct result is at least $O(1/\log (r-l))$. Then, we apply the Amplitude amplification method and the idea from Lemma \ref{lem:modified_hoyer03} that requires $O(\sqrt{\log(r-l)})$ iterations. Similarly, we find the maximal $d$ that finds a substring. This algorithm also performs $O(\sqrt{\log(r-l)})$ iterations due to \cite{ll2016,k2014}. The total complexity of the algorithm is $O(\sqrt{r-l}(\log(r-l))^{0.5(k-1)})$ due to the complexity of $\FindFrom{k}$.

\begin{lemma}
$\FindFixedPos{k}(l,r,t,s,left)$ returns the leftmost minimal substring $x[i,j]$ such that $sign(f(x[i,j]))\in s$ or NULL if there is no such substring. The expected running time is $O(\sqrt{r-l}(\log(r-l))^{0.5(k-1)})$. 
\end{lemma}
\begin{proof}
Let us show by induction that $\FindFrom{k}(l,r,t,d,s)$ returns the leftmost substring $x[i,j]$ such that $\sign(f(x[i,j]))\in s$. If $k=2$, we check whether $x_t=x_{t-1}$ before $x_t=x_{t+1}$.

Assume that there is another minimal substring $x[i',j']$ such that $i'\leq t\leq j'$, $f(x[i,j])=f(x[i',j'])$ and $i'<i$.

\begin{enumerate}
    \item Assume that there are $j_1$ and $i_2$ such that $i<j_1<i_2<j$, $|f(x[i,j_1])|=|f(x[i_2,j])|=k-1$ and $\sign(f(x[i,j_1]))=\sign(f(x[i_2,j]))\in s$.
    
    By induction one of the invocations of $\FindFrom{k-1}$ or $\FindFirst{k-1}$ finds $x[i_2,j]$ and it the leftmost. Therefore, $j'=j$. If $i'<i$, then $x[i',j']$ is not minimal or $|f(x[i',j'])|>|f(x[i,j])|$, a contradiction.
    
    \item Assume that there are $j_1$ and $i_2$ such that $i<i_2<j_1<j$, $|f(x[i,j_1])|=|f(x[i_2,j])|=k-1$ and $\sign(f(x[i,j_1]))=\sign(f(x[i_2,j]))\in s$.
     By induction $x[i_2,j]$ is the leftmost $\pm(k-1)$-substring. Therefore, $j'=j$. If $i'<i$, then $x[i',j']$ is not minimal or $|f(x[i',j'])|>|f(x[i,j])|$, a contradiction.
    
\end{enumerate}

If $d>r-l$ the algorithm finds $x[i,j]$. If  $d<r-l$, the algorithm could find the wrong substring (not the leftmost one containing $t$). So, we should to find the maximal $d$ such that $\FindFrom{k}$ finds a substring. In that case,  when we amplify the randomized version of the algorithm, we get the required one.

Searching by Grover's search for the maximal $d$ requires the same $O(\sqrt{r-l})$ expected number of iterations due to \cite{ll2016,k2014}.  
 The total complexity of the algorithm is $O(\sqrt{r-l}(\log(r-l))^{0.5(k-1)})$ due to the complexity of the $\FindFrom{k}$ procedure.
\end{proof}

$\FindFixedPos{k}(l,r,t,s,right)$ searches for the rightmost substring $x[i,j]$ such that $\sign(f(x[i,j]))\in s$ and $|f(x[i,j])|=k$, i.e. $i\leq t \leq j$ and there is no $x[i',j']$ such that $i'\leq t\leq j'$, $j<j'$ and $f(x[i',j'])=f(x[i,j])$.

The algorithm is similar to $\FindFixedPos{k}(l,r,t,s,left)$, but uses $\FindFromRight{k}$.

\subsection{\texorpdfstring{$\FindFirst{k}$}{FindFirst\_k} Algorithm's Description}

The $\FindFirst{k}$ procedure calls $\FindLeftFirst{k}$ or
$\FindRightFirst{k}$ depending on the direction. Since both version
are essentially symmetric, we only present the search from the left
below (i.e. when the direction is right).
For reasons that become clear in the proof, we need to boost the success
probability of some calls. We do so by repeating them several times and
taking the majority: by this we mean that we take the most common answer,
and return an error in case of a tie.

\begin{algorithm}
    \caption{$\FindRightFirst{k}(l,r,s)$. The algorithm for searching for the first $\pm k$-substring.\label{alg:substringk_first_right}}
    \begin{algorithmic}
        \State $lBorder\gets l, rBorder\gets r$
        \State $d\gets 1$ \Comment{depth of the search}
        \While{$lBorder+1<rBorder$}
            \State $mid \gets \lfloor(lBorder+rBorder)/2\rfloor $
            \State $v_l\gets \FindAny{k}(lBorder,mid,s)$
                \Comment{repeat $2d$ times and take the majority}
            \If{$v_l\neq \NULL$}
                \State $rBorder\gets mid$
            \EndIf
            \If{$v_l=\NULL$}
                \State $v_{mid}\gets \FindFixedPos{k}(lBorder,rBorder,mid,s,left)$
                    \Comment{majority of $2d$ runs}
                \If{$v_{mid}\neq \NULL$}
                    \State $v\gets v_{mid}$
                    \State Stop the loop.
                \EndIf
                \If{$v_{mid}=\NULL$}
                    \State $lBorder\gets mid+1$
                \EndIf
            \EndIf
            \State $d\gets d+1$
        \EndWhile            
        \State\Return{$v$}
    \end{algorithmic}
\end{algorithm}


\subsection{Proof of Claim on \texorpdfstring{$\FindFirst{k}$}{FindFirst\_k} Procedure from Proposition \ref{pr:findfrom}}
Let us prove the correctness of the algorithm for $direction=right$ and $s=\{+1\}$. The proof for other parameters is similar.

First, we show the correctness of the algorithm assuming there are no
errors.
The algorithm is essentially a binary search. At each step we find the middle of the search segment $[lBorder, rBorder]$ that is $mid = \lfloor(lBorder+ rBorder)/2 \rfloor$. There are three options.
\begin{itemize}
\item There is a $k$-substring in $[lBorder, mid]$, then the leftmost $k$-substring is in this segment. 
\item There are no $k$-substrings in $[lBorder, mid]$, but $mid$ is inside a $k$-substring. If we find the leftmost substring containing $min$, it is the required substring.
\item There are no $k$-substrings in $[lBorder, mid]$ and $mid$ is not inside a $k$-substring. Then the required substring is in $[mid+1, rBorder]$.
\end{itemize}
In each iteration of the loop the algorithm finds a smaller segment containing the leftmost $k$-substring or finds it if it contains $mid$. We find the $k$-substring in the iteration that corresponds  to the $[lBorder, rBorder]$ segment such that $(rBorder-lBorder)/2 \leq j-i$ or earlier.

Second, we compute complexity of the algorithm (taking into account
the repetitions and majority votes).
The $u$-th iteration of the loop considers a segment $[lBorder,rBorder]$. The length of this segment is at most $w\cdot 2^{-(u-1)}$ where $w=r-l$. The complexity of $\FindAny{k}(lBorder,mid,s)$ is at most $O\left(\sqrt{w\cdot 2^{-(u-1)-1}}\left(\log{(w\cdot 2^{-(u-1)-1})}\right)^{0.5(k-1)} \right)=O\left(\sqrt{w\cdot 2^{-(u-1)-1}}\left(\log{(r-l)}\right)^{0.5(k-1)} \right)$.
Also, $\FindFixedPos{k}(lBorder,rBorder,mid,s,left)$ has complexity
$O\left(\sqrt{w\cdot 2^{-(u-1)}}\left(\log{(w\cdot 2^{-(u-1)})}\right)^{0.5(k-1)} \right)=O\left(\sqrt{w\cdot 2^{-(u-1)}}\left(\log{(r-l)}\right)^{0.5(k-1)} \right)$. So the total complexity of the $u$-th iteration is $O\left(u\sqrt{w\cdot 2^{-(u-1)}}\left(\log{(r-l)}\right)^{0.5(k-1)} \right)$,
since at the $u$-th iteration, we repeat each call $2u$ times to take
a majority.
The number of iterations is at most $\log_2 w$. Let us compute the total complexity of the binary search part:

\begin{align*}
    O\left(\sum_{u=1}^{\log_2 w}2u\sqrt{w\cdot 2^{-(u-1)}}\left(\log{(r-l)}\right)^{0.5(k-1)}\right)
    &=O\left(\sqrt{w}\left(\log{(r-l)}\right)^{0.5(k-1)}\sum_{u=1}^{\log_2 w}u(\sqrt{2})^{-(u-1)}\right)\\
    &=O\left(\sqrt{w}\left(\log{(r-l)}\right)^{0.5(k-1)}\sum_{u=0}^{\infty }(u+1)(\sqrt{2})^{-u}\right)\\
    &=O\left(\sqrt{w}\left(\log{(r-l)}\right)^{0.5(k-1)}\frac{\sqrt{2}^2}{(\sqrt{2}-1)^2}\right)\\
    &=O\left(\sqrt{w}\left(\log{(r-l)}\right)^{0.5(k-1)}\right).
\end{align*}

Finally, we need to analyze the success probability of the algorithm:
at the $u^{th}$ iteration, the algorithm will run each test $2u$ times
and each test has a constant probability of failure $\varepsilon$.
Hence
for the algorithm to fail (that is make a decision that will not lead
to the first $\pm k$-substring) at iteration $u$, at least half of the
$2u$ runs must fail: this happens with probability at most
\[
    {2u\choose u}\varepsilon^u
    \leqslant \left(\frac{2ue}{u}\right)^u\varepsilon^u
    \leqslant (2e\varepsilon)^u.
\]
Hence the probability that the algorithm fails is bounded by
\[
    \sum_{u=1}^{\log_2 w}(2e\varepsilon)^u
    \leqslant\sum_{u=1}^{\infty}(2e\varepsilon)^u
    \leqslant \frac{2e\varepsilon}{1-2e\varepsilon}.
\]
By taking $\varepsilon$ small enough (say $2e\varepsilon<\tfrac{1}{3}$), which
is always possible by repeating the calls a constant number of times to
boost the probability, we can ensure that the algorithm a probability of
failure less than $1/2$.

\section{Proof of Theorem \ref{th:upper_bound}}
\label{sec:Dyck-n,k}

\begin{proof}
Let us show that if $x'$ contains $\pm(k+1)$-substring then one of  three conditions of $\dyck_{k,n}$ problem is broken.

Assume that $x'$ contains $(k+1)$ substring $x'[i,j]$. If $j\geq k+n$, then $f(x[i-k,n-1])>0$, because $f(x'[n,j])=j-n+1\leq k<k+1$. Therefore, prefix $x[0,i-k]$ is such that $f(x[0,i-k-1])<0$ or $f(x[0,n-1])>0$ because  $f(x[0,n-1])=f(x[0,i-k])+f(x[i-k-1,n-1])$. So, in that case we break one of conditions of $\dyck_{k,n}$ problem.

If $j<k+n$ then $x[i-k,j-k]$ is $(k+1)$ substring of $x$. 

Assume that $x'$ contains $-(k+1)$ substring $x'[i,j]$. If $i< k$, then $f(x[0,j-k])<0$, because $f(x'[i,k-1])=-(k-i)\geq -k>-(k+1)$ and $f(x[0,j-k])= f(x'[k,j])=f(x[i,j])-f(x[i,k-1])$. So, in that case the second condition of $\dyck_{k,n}$ problem is broken.

The complexity of Algorithm \ref{alg:dycknh} is the same as the complexity of $\FindAny{k+1}$ for $x'$ that is $O(\sqrt{n+2k}(\log(n+2k))^{0.5k})$ due to Proposition \ref{pr:findfrom}.

We can assume $n\geq 2k$ (otherwise, we can update $k\gets n/2$). Hence,
\[O(\sqrt{n+2k}(\log(n+2k))^{0.5k})=
O(\sqrt{2n}(\log(2n))^{0.5k})=O(\sqrt{n}(2\log{n})^{0.5k})=O(\sqrt{n}(\log{n})^{0.5k})\]

The error probability is the same as the complexity of $\FindAny{k+1}$. 
\end{proof}

\section{Reduction for the proof of Theorem \ref{t:dycklowerbound-power}}
\label{sec:reduction}

Before we describe the reduction in detail, we sketch the main idea. 
Recall that $\im(x)=|x|_0-|x|_1$. Note that \[\ex{2m}{m}{m+1}(x)=0 \iff \im(x)=2\] \[\ex{2m}{m}{m+1}(x)=1 \iff \im(x)=0\] whereas 
\begin{equation*}\dyck_{k,n}(x)=1\iff (\max_{p\text{ -- prefix of }x}{\im(p)}\leq k) \wedge (\min_{p\text{ -- prefix of }x}{\im(p)}\geq 0) \wedge (\im(x)=0).\end{equation*} If we could make sure that the minimum and maximum constraints are satisfied, $\dyck_{k,n}$ could be used to compute $\exnn$. To ensure the minimum constraint, we map each $0$ to $00$ and $1$ to $01$. However, this increases $\im(x)$ by $2m$ which can be fixed by appending $1^{2m}$ at the end. Importantly, the resulting sequence $x'$ has $\im(x')=\im(x)$. The first constraint (maximum over prefixes) can be fulfilled by having a sufficiently large $k$; $k=2m+3$ would suffice here. The same idea can be applied iteratively to $\ex{2m}{m}{m+1}$ where the inputs, which could now be the results of functions $\left(\ex{2m}{m}{m+1}\right)^{\ell-1}=x_i$, have been recursively mapped to sequences $x_i'$ with $\im(x_i')=\begin{cases}2 \text{ if }x_i=0\\0\text{ if }x_i=1\end{cases}$.

The reduction formally is as follows.

We call a string $B\in\left\{ 0,1\right\} ^{w}$ of even length
a $\left(w,h\right)$-sized block with width $w$ and height $h$
iff for any prefix $x$ of $B$: $0\leq \im(x)\leq h$
and either $\im(B)=0$ or $\im(B)=2$.

We establish a correspondence between inputs to $\left(\exnn\right)^\ell$ that satisfy the promise and  
$\left(w,h\right)$-sized blocks $B$
for appropriately chosen $w, h$,
so that $\left(\exnn\right)^\ell=1$
iff $\im(B)=0$.

For $l=0$ (the input bits), we have $0$ corresponding to a $(2,2)$-sized block of $00$ and $1$ to a $(2,2)$-sized block of $01$.

For $l>0$,
let us have input bits $x=(x_1, x_2, \ldots, x_{2m})$ of $\exnn$ satisfying the input promise. Assume that the bits (that could be equal to values of $\left(\exnn\right)^{\ell-1}$) correspond to $(w,h)$-sized blocks $B_1, B_2, \ldots, B_{2m}$. Define the sequence $B'=B_1B_2\ldots B_{2m}1^{2m}$. Then it is easy to verify the following claims:
\begin{enumerate}[1)]
    \item $B'$ is a $(2m(w+1),2(m+1)+h)$-sized block;
    \item The output bit of $\exnn(x)$ \emph{corresponds} to $B'$ because
    \[\im(B')=\sum_{i=1}^{2m}{\im(B_i)}+\im(1^{2m})=\begin{cases}2\text{ if }\exnn(x)=0\\ 0\text{ if }\exnn(x)=1\end{cases}.\]
\end{enumerate}

For $l=0$, the inputs correspond to $(2,2)$-sized blocks. Each level adds $2(m+1)$ to the height of the blocks reaching $2+2\ell(m+1)=O(m\ell)$. The width of blocks reaches $O((2m)^\ell)$.

Since for all $(w,h)$-sized blocks $B$: $\dyck_{h,w}(B)=1\iff \im(B)=0$ one can solve the $\left(\exnn\right)^\ell$ problem by running $\dyck_{h,w}$ on the \emph{corresponding} block.

See Figure \ref{fig:ex-dyck-reduction}.

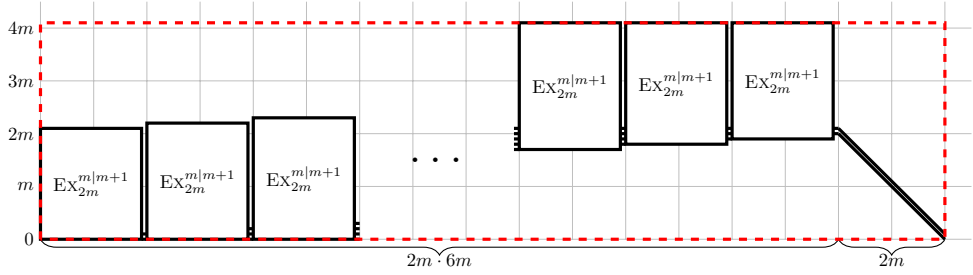
\begin{figure}[H]
\begin{centering}
\begin{tikzpicture}[scale=0.7,transform shape]
\tikzstyle{invis} = [outer sep=0,inner sep=0,minimum size=0]
\tikzstyle{vt} = [very thick]
\draw [step=1cm, very thin,lightgray] (1,0) grid (18.5,4.5) node (v2) {};

\node[invis,label=left:{$0$}] at (1,0) {};
\node[invis,label=left:{$m$}] at (1,1) {};
\foreach \x in {2,...,4}
{
	\node[invis,label=left:{$\x m$}] at (1,\x) {};
}

\draw [vt](16,2) -- (18,0);
\draw [vt](16,2.1) -- (18,0.1);

\draw [decorate, decoration={brace, mirror, amplitude=6}](1,0) -- (16,0) node[invis,midway,yshift=-0.4cm]{$2m\cdot 6m$};
\draw [decorate, decoration={brace, mirror, amplitude=6}](16,0) -- (18,0) node[invis,midway,yshift=-0.4cm]{$2m$};

\draw[vt,fill=white](1,0) rectangle (2.9,2.1) node[midway] {$\exnn$};

\draw[vt,fill=white](3,0) rectangle (4.9,2.2) node[midway] {$\exnn$};

\draw[vt,fill=white](5,0) rectangle (6.9,2.3) node[midway] {$\exnn$};

\node at (8.5,1.5) {\Huge{$\dots$}};g

\draw[vt,fill=white](10,1.7) rectangle (11.9,4.1) node[midway] {$\exnn$};
\draw[vt,fill=white](12,1.8) rectangle (13.9,4.1) node[midway] {$\exnn$};
\draw[vt,fill=white](14,1.9) rectangle (15.9,4.1) node[midway] {$\exnn$};

\draw [vt](2.9,0) -- +(0.1,0);
\draw [vt](2.9,0.1) -- +(0.1,0);

\draw [vt](4.9,0) -- +(0.1,0);
\draw [vt](4.9,0.1) -- +(0.1,0);
\draw [vt](4.9,0.2) -- +(0.1,0);

\draw [vt](6.9,0) -- +(0.1,0);
\draw [vt](6.9,0.1) -- +(0.1,0);
\draw [vt](6.9,0.2) -- +(0.1,0);
\draw [vt](6.9,0.3) -- +(0.1,0);

\draw [vt](15.9,2) -- +(0.1,0);
\draw [vt](15.9,2.1) -- +(0.1,0);

\draw [vt](13.9,1.9) -- +(0.1,0);
\draw [vt](13.9,2) -- +(0.1,0);
\draw [vt](13.9,2.1) -- +(0.1,0);

\draw [vt](11.9,1.8) -- +(0.1,0);
\draw [vt](11.9,1.9) -- +(0.1,0);
\draw [vt](11.9,2) -- +(0.1,0);
\draw [vt](11.9,2.1) -- +(0.1,0);

\draw [vt](9.9,1.7) -- +(0.1,0);
\draw [vt](9.9,1.8) -- +(0.1,0);
\draw [vt](9.9,1.9) -- +(0.1,0);
\draw [vt](9.9,2) -- +(0.1,0);
\draw [vt](9.9,2.1) -- +(0.1,0);

\draw[vt,red, dashed](1,0) rectangle (18,4.1);

\end{tikzpicture}
\par\end{centering}
\caption{\label{fig:ex-dyck-reduction}The reduction $\exnn\circ\exnn\protect\leqslant\dyck_{4m+6,12m^{2}+2m}$.
The line of the graph follows the input word along the \emph{x}-axis and shows
the number of yet-unclosed parenthesis along the \emph{y}-axis  (i.e., a zoomed-out version of Figure \ref{f:pyramids}).
The input word $B_1 B_2 \dots B_{2m} 1^{2m}$ corresponds to the outer function $\exnn$ with $B_j$
being a block corresponding to the output of an inner $\exnn$.
The ticks at the starts and ends of blocks depict that if the line enters the block at height $i$, it exits at height $i$ or $i+2$. In the block the line never goes below $0$ or above $h+i$.
The red dashed part then forms a new block $B'$.
By replacing the blocks $B_j$ with blocks $B'$ we can further iterate $\exnn$ to get the reduction $\exnn\circ\left(\exnn\right)^{\ell-1}\protect\leqslant\dyck_{O\left(\ell m\right),O\left(\left(2 m \right)^{\ell}\right)}$.}
\end{figure}

\section{A quantum algorithm for \texorpdfstring{$\dirtwod_{n,k}$}{2D-DConnectivity}}
\label{sec:narrowgrid}

In this section, we prove Theorem \ref{th:upper_dirtwod} by  
constructing a quantum algorithm for $\dirtwod_{n,k}$. 
The main idea is to construct an AND-OR formula
for $\dirtwod_{n,k}$ and to use one of quantum algorithms for AND-OR formula evaluation.
To achieve the optimal query complexity, we use the algorithm by Reichardt \cite{Reichardt14} which evaluates an AND-OR formula of size $L$ with $O(\sqrt{L})$ queries.
To achieve a time efficient quantum algorithm,
we can use quantum algorithms from \cite{Ambainis+10} or \cite{Reichardt11}
for which the number of queries is slightly larger ($O(\sqrt{Ld})$ for \cite{Ambainis+10} 
and $O(\sqrt{L \log L})$ for \cite{Reichardt11})
and the number of non-query steps is $O(\log^c L)$ per one query step. For the formula that we construct, $d=\log L$ and either of those 
quantum algorithms uses $O(\sqrt{L \log L})$ queries and $O(\sqrt{L} \log^c L)$ time steps.

We first deal with the case when $n=2^m$ for some non-negative integer $m$. The idea for the construction of the AND-OR formula is to split the grid in two: any path from $(0,0)$ to $(n,k)$ must pass
through a vertex $(\frac{n}{2},r)$ for some $r:\;1\leq r\leq k$. For the paths
to and from $(\frac{n}{2},r)$ we can apply this reasoning recursively. Let us denote by
$F_{\mu,\kappa,i,j}$ our formula for the path from vertex $(i,j)$ to 
$(i+2^\mu,j+\kappa)$, and by $L_{\mu,\kappa}$ its size (the number of variable instances
it has; it does not depend on $i$, $j$). Thus we have the recurrent formulae
$$
F_{\mu,\kappa,i,j}=
\bigvee_{r=0}^\kappa\left(F_{\mu-1,r,i,j}\land F_{\mu-1,\kappa-r,i+2^{\mu-1},j+r}\right),
$$
$$
L_{\mu,\kappa}=
\sum_{r=0}^\kappa\left(L_{\mu-1,r} + L_{\mu-1,\kappa-r}\right)
=2\sum_{r=0}^\kappa L_{\mu-1,r}.
$$
For the base case $F_{0,\kappa,i,j}$ (i.~e.\ for a $1\times \kappa$ grid) we
simply use an OR of all the paths (represented as an AND of all its edges). There are $\kappa+1$ 
paths, each of length $\kappa+1$, thus $L_{0,\kappa}=(\kappa+1)^2$.

It follows by induction on $\mu$ that 
$L_{\mu,\kappa} < 2^{\mu+1}\cdot\binom{\kappa+\mu+2}{\kappa}$. For the induction basis 
we have $L_{0,\kappa} < (\kappa+1)(\kappa+2) = 2 \binom{\kappa+2}{\kappa}$, and
for the induction step:
$$
L_{\mu,\kappa}=2\sum_{r=0}^\kappa L_{\mu-1,r}
< 2^{\mu+1}\sum_{r=0}^\kappa \binom{r+\mu+1}{r}=2^{\mu+1}\binom{\kappa+\mu+2}{\kappa}.
$$
Using a well-known upper bound for binomial coefficients we obtain: 
$L_{m,k}<2^{m+1}(e\cdot(k+m+2)/k)^k=O\left(n(e(1+\frac{\log_2 n}{k}))^k\right)$.
There exists a quantum algorithm with $O(\sqrt{L})$ queries for a formula of size $L$
\cite{Reichardt14}, thus we obtain the complexity mentioned in the theorem statement. 

For an arbitrary $n$ we can find the smallest $m$ for which $n\leq 2^m$ and use the formula for the $2^m \times k$
grid obtained by adding ancillary edges from the vertex $(n,k)$ to $(2^m,k)$
(using the edge variables of the added part of the grid as constants).
Since the value of $n$ thus increases no more than two times, the complexity
estimation increases by at most a constant multiplier.

\end{document}